\documentclass{lmcs}
\pdfoutput=1
\usepackage[utf8]{inputenc}

\usepackage{lastpage}
\lmcsdoi{20}{3}{27}
\lmcsheading{}{\pageref{LastPage}}{}{}%
{Apr.~17,~2023}{Sep.~30,~2024}{}

\pdfoutput=1
\usepackage{macros}

\keywords{Kleene algebra with tests; closed languages; relation algebra; guarded strings}

\begin{document}

\title[Completeness Theorems for Kleene algebra with tests and top]%
{Completeness Theorems\texorpdfstring{\\}{} for Kleene Algebra with tests and top}%
\titlecomment{%
  This paper is an extended version of the paper with a similar title which appeared in Proc.\ CONCUR'22~\cite{pw:concur22:katop}; we summarise the additions at the end of the introduction.
}
\thanks{%
  This work has been supported by the ERC (CoVeCe, grant No 678157), by the LABEX MILYON (ANR-10-LABX-0070), within the program ANR-11-IDEX-0007, and by the project ‘Mode(l)s of Verification and Monitorability’ (MoVeMent) (grant No 217987) of the Icelandic Research Fund.
}

\author[D.~Pous]{Damien Pous\lmcsorcid{0000-0002-1220-4399}}[a]
\author[J.~Wagemaker]{Jana Wagemaker\lmcsorcid{0000-0002-8616-3905}}[b]

% affiliation 1 (automatically numbered a)
\address{Plume team, CNRS, LIP, ENS de Lyon}
\email{Damien.Pous@ens-lyon.fr}

% affiliation 2 (automatically numbered b)
\address{iCIS, Department of Computer Science, Radboud University}
\email{jana.wagemaker@ru.nl}

% \ccsdesc[500]{Theory of computation~Equational logic and rewriting}
% \ccsdesc[500]{Theory of computation~Logic and verification}
% \ccsdesc[500]{Theory of computation~Regular languages}

\begin{abstract}
  We prove two completeness results for Kleene algebra with tests and a top element, with respect to guarded string languages and binary relations.
  While the equational theories of those two classes of models coincide over the signature of Kleene algebra, this is no longer the case when we consider an additional constant ``top'' for the full element. Indeed, the full relation satisfies more laws than the full language, and we show that those additional laws can all be derived from a single additional axiom.
  We recover that the two equational theories coincide if we slightly generalise the notion of relational model, allowing sub-algebras of relations where top is a greatest element but not necessarily the full relation.

  We use models of closed languages and reductions in order to prove our completeness results, which are relative to any axiomatisation of the algebra of regular events.

  For one of our constructions, we extend the concept of finite monoid recognisability to guarded-string languages; this device makes it possible to obtain a \pspace algorithm for the equational theory of binary relations.
\end{abstract}

\maketitle

\section{Introduction}
\label{sec:intro}

The axiomatic treatment of regular expressions and languages was developed extensively by Conway~\cite{conway71}, after earlier work of Kleene~\cite{Kle56}. Conway asked a difficult question: how to axiomatise the equations between regular expressions that hold under their standard interpretation as formal languages? Redko had proved that every purely equational axiomatisation must be infinite~\cite{redko1964algebra}. Conway proposed such an infinite axiomatisation, which Krob proved to be complete twenty years later---1991~\cite{Krob91a}. Conway had also proposed finite quasi-equational axiomatisations, one of which Kozen proved to be complete also in 1991~\cite{Kozen91}---this axiomatisation is now commonly called \emph{Kleene algebra}. By an additional remark of Boffa~\cite{Boffa90}, this latter completeness result can also be obtained as a consequence of Krob's completeness result. In the end, all finite quasi-equational axiomatisations proposed by Conway, as well as a few other ones, are actually complete~\cite{Krob91a,Boffa95}.

In symbols, writing $\plang e$ for the language of a regular expression $e$ and $\KA\proves e=f$ when the equation $e=f$ is derivable in any of the aforementioned axiomatisations, we have that for all regular expressions $e,f$,
\begin{align*}
  \KA\proves e=f \iff \plang e=\plang f
\end{align*}
The above equivalence extends with two more clauses. When an equation is derivable, % ($\KA\proves e=f$),
it must hold in all models of the chosen axiomatisation.
These include in particular language models (\LANG) and relational models (\REL), for which we actually have an equivalence: writing $\X\models e=f$ when the equation $e=f$ holds in all members of a class of models $\X$, we actually have:
\begin{gather*}
  \KA\proves e=f
  \iff
  \REL\models e=f
  \iff
  \LANG\models e=f
  \iff
  \plang e = \plang f
\end{gather*}
Completeness w.r.t. \LANG is immediate given the previous equivalence: the language interpretation of a regular expression lies in \LANG. This is less obvious for \REL: completeness comes from a nice trick due to Pratt showing that every member of \LANG embeds into a member of \REL~\cite[third page]{pratt80:cayley}.

As an immediate consequence of the above equivalence, the equational theory of \REL (or \LANG) is decidable---more precisely, \pspace-complete.

\medskip

The above-mentioned results apply to the \emph{regular} operations and constants: composition, union, Kleene star, identity, emptiness.
A natural question is whether they extend to other operations or constants, such as intersection, converse, fullness.
The case of converse was dealt with by Ésik et al.: the equational theories of \REL and \LANG differ in the presence of converse but both can be axiomatised~\cite{BES95,EB95}, and they remain \pspace-complete~\cite{bp:jlamp16:kac}.
The case of intersection (with or without converse or the various constants) is significantly more difficult, and remains partly open, see~\cite{AMN11,bp:lics15:paka,Nakamura17,dp:concur18:kl}.
In this paper we focus on the addition of a constant $\top$, interpreted as the full language in \LANG and as the full relation in \REL.

The usefulness of adding such a constant was demonstrated recently in the context of Kleene algebras with tests (KAT)~\cite{kozens96:kat:completeness:decidability,cohenks96:kat:complexity,kozen:97:kat}, to model \emph{incorrectness logic}~\cite{ZhangAG22}. Indeed, while KAT alone makes it possible to model Hoare triples for partial correctness~\cite{kozen00:kat:hoare}, the addition of a full element makes it possible to compare the (co)domains of relations, and thus to encode \emph{incorrectness triples}~\cite[Section~5.3]{OHearn20}. KAT with a top element was also used earlier, as an intermediate structure to characterise a semantics for abnormal termination~\cite[Definition~12]{Mamouras17}.

The theory of KAT was developed extensively by Kozen et al., and it is very similar to that of Kleene algebra. For instance, the above equivalences extend to
\begin{gather*}
  \KAT\proves e=f
  \iff
  \REL\models e=f
  \iff
  \GSL\models e=f
  \iff
  [e] = [f]
\end{gather*}
There, \GSL is the generalisation of language models (\LANG) to \emph{guarded string} language models, and $[e]$ denotes the guarded string language interpretation of a regular expression with tests $e$.
Here also, we get decidability in \pspace, using an appropriate generalisation of finite automata on words to finite automata on guarded strings~\cite{cohenks96:kat:complexity,K17a}.

Dealing with Kleene algebra with tests has important applications in program verification: they make it possible to represent and reason about the big-step semantics of while programs, algebraically. This was used for instance to analyse compiler optimisations~\cite{kozenp00:kat:compiler:opts}. The decidability result was also implemented in proof assistants such as Coq and Isabelle/HOL, in order to automate some reasoning steps about binary relations and Hoare logic on while programs~\cite{pous:itp13:ra,KraussN12}.

\medskip

Like in~\cite{ZhangAG22} the problem we consider in this paper is that of adding a constant $\top$.
% to Kleene algebra with tests.
%
As expected, one should consider an axiom expressing that $\top$ is a greatest element:
\begin{equation}
  \label{ax:T}
  \tag{T}
  x \leq \top
\end{equation}
(Where $x\leq y$ is a shorthand for $x+y=y$.)
Together with the Kleene algebra axioms, axiom~(\ref{ax:T}) yields a complete axiomatisation w.r.t.\ language models: we sketched a proof for the case without tests in~\cite[Example~3.4]{prw:ramics21:mkah}, which we make fully explicit here in Section~\ref{sec:lang} (Theorem~\ref{thm:lang}). This proof gives us as a byproduct that the equational theory of Kleene algebras (with tests and) with a greatest element remains \pspace-complete.

Unfortunately, the previous axiom is not enough to deal with relational models. In fact, in the presence of $\top$, the equational theories of \LANG (or \GSL) and \REL differ.
Indeed, there are laws such as $\top x\top y\top = \top y\top x\top$~\cite[page~13]{pous:stacs18}, or $\top x\top x = \top x$~\cite[page~14]{ZhangAG22}, which are valid in \REL, but not in \LANG.
% In other words, the equational theories of \LANG and \REL differ in the presence of $\top$.

In the present paper, we show that it suffices to further add the following axiom in order to obtain a complete axiomatisation for \REL (Theorem~\ref{thm:rel}):
\addtocounter{equation}{1}% to avoid hyperref complaints
\begin{equation}
  \label{ax:F}
  \tag{F}
  x \leq x\cdot \top \cdot x
\end{equation}
This inequation is mentioned in~\cite[page~14]{ZhangAG22}; it holds in relational models, but not in language ones.
Thanks to (\ref{ax:T}), axiom (\ref{ax:F})
may be seen as a consequence of Ésik et al.'s axiom $x\leq x\cdot x^\circ\cdot x$ for dealing with converse $(\cdot^\circ)$ in relational models~\cite{EB95,BES95}.
How to use axiom~(\ref{ax:F}) in an equational proof is not so intuitive: it does not give rise to a natural notion of normal form, and it must often be used in conjunction with (\ref{ax:T}) in order to compensate for the fact that it duplicates subterms. For instance here is how we can prove the first of the aforementioned laws:
\begin{align*}
  \tag{by~(\ref{ax:F})}
  \top\pdot x\pdot\top\pdot y\pdot\top
  &\leq \top\pdot x\pdot\top\pdot y\pdot\top \pdot ~\top~ \pdot \top\pdot x\pdot\underline{\top\pdot y\pdot\top} \\
  \tag{by~(\ref{ax:T})}
  &\leq \top\pdot x\pdot\top\pdot y\pdot\underline{\top \pdot \top \pdot \top}\pdot x\pdot\top\\
  \tag{by~(\ref{ax:T})}
  &\leq \underline{\top\pdot x\pdot\top}\pdot y\pdot\top\pdot x\pdot\top\\
  \tag{by~(\ref{ax:T})}
  &\leq \top\pdot y\pdot\top\pdot x\pdot\top
\end{align*}
(We wrote compositions by juxtaposition, skipped the associativity steps, and underlined the subterms to be simplified by axiom~(\ref{ax:T})---the converse inequation is derived symmetrically.)
Our completeness proof actually goes via a factorisation property (Proposition \ref{prop:hom:cf}) intuitively asserting that one can always proceed in this way to reason about star-free expressions: expand the expressions using (\ref{ax:F}) a number of times, then remove spurious subterms using~(\ref{ax:T}). Combining such a technique together with Kleene algebra reasoning for star is the second challenge we address in the present work.

To get a grasp on the difficulties, the reader may try to find a proof of the following valid law of \REL, using \KA and axiom~(\ref{ax:F}):
\addtocounter{equation}{1}% to avoid hyperref complaints
\begin{equation}
  \label{eq:ex}
  \tag{$\star$}
  (aaa)^* \leq (aaa)^*\top (aa)^* + (aa)^*a\top (aaa)^*
\end{equation}
(Note that there is an easy proof using \KA and axiom~\eqref{ax:T}, which however breaks with the following variant: $b(aaa)^*b \leq b((aaa)^*b\top b(aa)^* + (aa)^*ab\top b(aaa)^*)b$.)
We give a solution to this exercise in Section~\ref{ssec:ex}.

Finally, we show that the difference between the equational theories of language and relational models can be blurred if we slightly generalise the notion of relational model, allowing $\top$ to be any greatest relation rather than the full one\footnote{Considering subalgebras where only certain relations are kept, since otherwise the only greatest relation is the full one.} (Corollary~\ref{cor:relp}).

\medskip

We prove our two main theorems using the concept of \emph{closed language model} for Kleene algebra with hypotheses~\cite{dkpp:fossacs19:kah}, and the reduction technique made explicit in~\cite{prw:ramics21:mkah,KappeB0WZ20}\footnote{Such a technique is somehow implicit in Kozen and Smith's completeness proof for KAT~\cite{kozens96:kat:completeness:decidability} and Ésik et al.'s completeness proof for Kleene algebra with converse~\cite{BES95,EB95}.}.
Intuitively, we establish reductions from KAT with (\ref{ax:T}) and KAT with (\ref{ax:T}, \ref{ax:F}) to plain KAT, so that we can deduce completeness and decidability of the former theories from completeness and decidability of the latter one.

While the first reduction is relatively straightforward---this is a syntactical linear reduction, the second one is not. We exploit the aforementioned factorisation result (Proposition~\ref{prop:hom:cf}) and the theory of language recognition by finite monoids~\cite{eilenberg1974automata,sakarovitch_2009} in order to show that regular languages are preserved by a certain closure operation, and that this preservation property can be justified algebraically (Proposition~\ref{prop:s:props}). Doing so for Kleene algebra with tests first requires us to extend such ideas to deal with guarded string languages (Section~\ref{sssec:gmonoids}).
Moreover, in order to establish the correspondence between the closed languages used there and relational models, we resort to a graph theoretical characterisation of the equational theory of \REL~(Theorem~\ref{thm:rel:graphs}, whose main ingredients date back to the works of Freyd and Scedrov~\cite[page 208]{FreydScedrov90} and Andréka and Bredikhin~\cite[Theorem~1]{AB95}).

\paragraph*{Differences with the version in Proc.\ CONCUR'22}

The three main differences with~\cite{pw:concur22:katop} are the following:
1/ we deal with Kleene algebra with tests rather than plain Kleene algebra; 2/ we use a technique based on finite monoids rather than on finite automata to deal with axiom~(\ref{ax:F}) in Section~\ref{ssec:rel:compl:clang}; 3/ we provide a \pspace algorithm for the equational theory, a problem which we had left open.

For 1/ we need to move from word languages to guarded string languages. This change is pervasive but the development scales smoothly without the need for major conceptual changes. For instance, when it comes to the graph-theoretical characterisation (Theorem~\ref{thm:rel:graphs}), it suffices to add labels to graph vertices. In fact, once the definitions are properly extended, the proofs remain almost unchanged.

Contribution 2/ is entirely new. It makes it possible to avoid the slightly cumbersome automata construction we were using in~\cite[Section~4.2]{pw:concur22:katop} (which would have been slightly painful to extend to guarded string automata). It is this new construction based on monoids which lead us to the \pspace algorithm we present in Section~\ref{ssec:alg} (Contribution 3/). While this idea is conceptually simple once we know the tools from monoid-based recognition of regular languages, some care is required here since we deal with guarded string languages: we have to develop the premises of a theory of monoid-based recognition of guarded string languages. This is why we first explain the idea on plain regular languages (Section~\ref{sssec:interlude}).

Our results readily apply to Kleene algebra with top but without tests (cf. Remark~\ref{rem:singleatom} below). Nevertheless the reader not interested in tests and guarded strings may find it easier to read~\cite{pw:concur22:katop} except for its section~4.2, and then jump to Sections~\ref{sssec:interlude}, \ref{ssec:ex} and~\ref{ssec:alg} from the present paper.

\paragraph*{Outline}

We setup and recall basic notation for regular expressions, formal languages and universal algebra in Section~\ref{sec:prelim}. Then we deal with guarded string language models in Section~\ref{sec:lang}, and relational models in Section~\ref{sec:rel}. While the language case was already sketched in~\cite[Example~3.4]{prw:ramics21:mkah} (without tests), we find it useful to treat it explicitly here, before dealing with the more involved case of relations: it illustrates the reduction method in a simpler setting, and we build on the reduction for languages to establish the reduction for relations.
We finally prove completeness with respect to generalised relational models in Section~\ref{sec:relp}.

\section{Preliminaries}
\label{sec:prelim}

Given a set $X$, we write $X^*$ for the set of \emph{words} over $X$: finite sequences of elements of $X$. We let $u,v$ range over words, we write $\epsilon$ for the empty word, $uv$ for the concatenation of two words $u,v$, and $u^i$ for the concatenation of $i$ copies of a word $u$.
We let $e,f$ range over \emph{regular expressions over $X$}, generated by the following grammar:
\begin{align*}
  e,f &::= e+f \mid e\cdot f \mid e^*\mid 0 \mid 1 \mid x\in X
\end{align*}
We sometimes omit the dots in regular expressions, writing, e.g., $ab^*$ for $a\cdot b^*$.
A \emph{language} is a set of words.
As usual, we associate a language $\plang e$ to every regular expression $e$, the \emph{language of $e$}.

We fix a finite set $\Sigma$ of \emph{letters}, ranged over using $a,b$, and a finite set $\Atom$ of \emph{atoms}, ranged over using $\alpha,\beta,\gamma$.
We write $\Sigmat$ for the set $\Sigma$ extended with a new element $\top$ called \emph{top}.

We call \emph{expressions} the regular expressions over $\SigmatA$.
We call \emph{plain regular expressions} the regular expressions over $\Sigma$.
We shall sometimes see words over $\SigmatA$ as expressions. E.g., the word $\alpha a \beta\top\gamma$ can be seen as the expression $\alpha\cdot a\cdot\beta\cdot\top\cdot\gamma$.

We consider signatures $S\eqdef\set{{+}_2,{\cdot}_2,{\cdot^*}_1,0_0,1_0}\cup\set{\alpha_0\mid \alpha\in\Atom}$ and $\St\eqdef S\cup\set{\top_0}$.
In those signatures, there is a constant symbol for every atom $\alpha\in\Atom$.

Expressions form the free $\St$-algebra over $\Sigma$.
Given an $\St$-algebra $A$ and a valuation $\sigma\colon\Sigma\to A$, we write $\hat\sigma$ for the unique homomorphism extending $\sigma$ to expressions.
Note that $\hat\sigma(\top)=\top\!_A$ and $\hat\sigma(\alpha)=\alpha\!_A$ by definition: top and atoms are constants, not variables.

Given a class $\X$ of $\St$-algebras and two expressions $e,f$, we write $\X\models e=f$ if for all members $A$ of $\X$ and all valuations $\sigma\colon\Sigma\to A$, we have $\hat\sigma(e)=\hat\sigma(f)$.

An \emph{equation} is a pair of expressions $e,f$, written $e=f$. We write $e\leq f$, an \emph{inequation}, as a shorthand for the equation $e+f=f$.
An \emph{axiomatisation} is a set of equations (or implications between equations). Given such a set $\E$, we write $\E\proves e=f$ when the equation $e=f$ is derivable from $\E$ using the rules of equational reasoning (where letters from $\Sigma$ appearing in the equations of $\E$ can be substituted by arbitrary terms).

\medskip

We let \KA stand for any axiomatisation over plain regular expressions which is sound and complete w.r.t. the regular language interpretation, i.e., such that for all plain regular expressions $e,f$, we have\footnote{Actually, we require slightly more if the axiomatisation contains implications: those implications should be valid in the models of languages and binary relations.}
\begin{equation*}
  % \label{eq:KA}
  % \tag{$RM\dagger$}
  \KA\proves e=f \iff \plang e=\plang f
\end{equation*}
As explained in the introduction, valid candidates for \KA include Conway's infinite but purely equational axiomatisation~\cite[page~116]{conway71} (proved complete by Krob~\cite{Krob91a}), Kozen's Kleene algebras~\cite{Kozen91}, left-handed Kleene algebras~\cite{KozenS12,ddp:lpar18:lefthanded}, and Boffa's algebras~\cite{Boffa95}.

Also note that the above requirement is equivalent to the following one, since $L\subseteq K$ iff $L\cup K=K$ for all languages $L,K$:
\begin{equation*}
  % \label{eq:KA'}
  % \tag{$RM\ddagger$}
  \KA\proves e\leq f \iff \plang e\subseteq \plang f
\end{equation*}
Let \KAT, \emph{Kleene algebra with tests}, be the union of \KA and the following equations:
\begin{align}
  \label{ax:atoms}
  \tag{A}
  \sum_{\alpha\in\Atom}\alpha=1 &&
  \alpha\cdot\beta = 0\quad(\forall \alpha\neq\beta)
\end{align}
Note that we can deduce $\KAT\proves\alpha\cdot\alpha = \alpha$ for all atoms $\alpha$, by neutrality of $1$ and distributivity.

A \emph{guarded string (over an alphabet $X$)} is a word over $\XA$ starting with an atom, alternating between atoms and elements in $X$, and ending with an atom. The \emph{length} of a guarded string is the number of $X$-elements in it. For instance, $\alpha$, $\alpha x \beta$, $\alpha x \beta y \gamma$ are guarded strings of respective lengths 0, 1, and 2, when $x,y\in X$. We write $\GS_X$ for the set of guarded strings over $X$, which is in bijection with $(\Atom\times X)^*\times \Atom$.

The \emph{coalesced product} is a partial binary operation on guarded strings: if $u=u'\alpha$ and $v=\beta v'$ are two guarded strings, then their coalesced product $u\diamond v$ is defined if $\alpha=\beta$, in which case $u\diamond v\eqdef u'\alpha v'$.

A \emph{guarded string language} is a set of guarded strings.
The \emph{guarded string language of an expression $e$} is defined as
\begin{align*}
[e]\eqdef \gs\plang e
\end{align*}
where $\gs$ is the following function from languages over $\SigmatA$ to guarded string languages over $\Sigmat$:
\begin{align*}
  \gs\colon &\pow((\SigmatA)^*)\to \pow(\GS_\Sigmat)\\
  &L \mapsto \set{\alpha_0 a_0 \dots a_{n-1}\alpha_n\mid \exists i_0,\dots,i_n\in\NN,~\alpha^{i_0}_0 a_0 \dots a_{n-1}\alpha^{i_n}_n\in L}
\end{align*}
For instance, if $\Atom$ consists of three distinct atoms $\alpha,\beta,\gamma$, then
\begin{align*}
  [(\alpha a \beta + \beta b \gamma)^*] = \set{\alpha,\beta,\gamma,\alpha a\beta, \beta b\gamma, \alpha a\beta b\gamma}
\end{align*}

In the absence of top, \KAT is sound and complete w.r.t. the guarded string language interpretation: for all expressions without top,
we have
\begin{align}
  \label{eq:KAT}
  \tag{$\dagger$}
  \KAT\proves e=f &\iff [e]=[f] \\
  \label{eq:KAT'}
  \tag{$\ddagger$}
  \KAT\proves e\leq f &\iff [e]\subseteq [f]
\end{align}
(This is essentially~\cite[Theorem~8]{kozens96:kat:completeness:decidability}, even though we work with an abstract version of \KAT here, cf Remark~\ref{rem:std:kat} below.)
For expressions without top, it is also known that \KAT is sound and complete with respect to both relational and guarded string language models~\cite[Theorem~7]{kozens96:kat:completeness:decidability}, which we define in the following sections. The point of this work is to deal with the constant top.

\begin{rem}
  \label{rem:singleatom}
  By choosing a singleton set $\set{\ast}$ for $\Atom$, we recover the case of plain Kleene algebra with top we covered in~\cite{pw:concur22:katop}.
  Indeed, in that case, the first axioms in~\eqref{ax:atoms} reads as $\ast=1$ so that atoms become redundant in the syntax of expressions, guarded strings over $X$ are in one-to-one correspondence with words over $X$, and accordingly, guarded string languages are just standard word languages.
\end{rem}

\begin{rem}
  \label{rem:std:kat}
  In the literature, KAT is usually presented as a two-sorted system: one sort for \emph{programs} forming a Kleene algebra, and one sort for \emph{tests} forming a Boolean algebra which embeds as a lattice into the former sort. Given a finite set $T$ of test variables, every element of the free Boolean algebra over $T$ can be represented as a disjunction of atoms in $\Atom\eqdef 2^T$. This idea makes it possible to normalise standard KAT expressions in such a way that all tests are atoms, giving rise to the syntax and axioms we use in the present paper.
  This idea is already there in the completeness proof of KAT~\cite{kozens96:kat:completeness:decidability}; axioms~(\ref{ax:atoms}) appear explicitly in works about NetKAT~\cite[Figure~6]{AndersonFGJKSW14:netkat}; we gave a formal reduction between the two presentations in~\cite[Section~4.2]{prw:ramics21:mkah}.
  We prefer this setup because it is single-sorted and slightly more abstract (e.g., we could imagine models where the set of atoms is not of the form $2^T$.)
  Note that when there are no tests in standard KAT (i.e., $T=\emptyset$), we fall back into the special case discussed in Remark~\ref{rem:singleatom}: $2^\emptyset$ is a singleton.
\end{rem}

\section{Guarded string language models}
\label{sec:lang}

Let $X$ be an alphabet and let $L,K$ range over guarded string languages on the alphabet $X$.
These form an $\St$-algebra with the operations defined as follows:
\begin{align*}
  L+K &\eqdef L\cup K&
  0 &\eqdef \emptyset&
  \top &\eqdef \GS_X\\
  L\cdot K &\eqdef \set{u\diamond v\mid u\in L~\land~v\in K}&
                                                              1 &\eqdef \set{\alpha\mid \alpha\in\Atom}&
  \alpha &\eqdef \set{\alpha}\\
  L^* &\eqdef \bigcup_{n\in\NN}L^n & L^0&\eqdef 1 & L^{i+1}&\eqdef L\cdot L^i
\end{align*}
(That is, $+$ is set-theoretic union, $0$ and $\top$ are the empty and full languages, respectively, $\cdot$ is guarded string language concatenation, via coalesced product, $1$ contains all guarded strings consisting of a single atom, and $\cdot^*$ is obtained via iteration.)
We write \GSL for the class of all $\St$-algebras of the above shape.

\begin{fact}
  The guarded string language interpretation of expressions, $[\cdot]=\gs\plang\cdot$, is the unique $S$-algebra homomorphism satisfying $[a]=\set{\alpha a \beta\mid\alpha,\beta\in\Atom}$ for all $a\in\Sigmat$.
  This is not an $\St$-algebra homomorphism, since $[\top]=\set{\alpha \top \beta\mid\alpha,\beta\in\Atom}\subsetneq\GS_\Sigmat=\top$.
\end{fact}

Let \KATT, \emph{KAT with a top element}, denote the union of the axioms from \KAT and axiom~(\ref{ax:T}). We prove in this section that \KATT is sound and complete for \GSL.
Following the strategy from~\cite{dkpp:fossacs19:kah,prw:ramics21:mkah}, the first step consists of defining the closure operation below, according to the axiom~(\ref{ax:T}) we add to \KAT:
\begin{defi}[Language closure $\CT$]
  \label{def:ct}
  Given two guarded strings $u,v$ over $\Sigmat$, we write $u\rT v$ if $u=l\diamond w\diamond r$ and  $v=l\top r$ for some guarded strings $l,w,r$.
  Given a guarded string language $L$ over $\Sigmat$, we call \emph{$T$-closure of $L$} the following guarded string language
  \begin{align*}
    \CT(L)\eqdef\set{u\mid u\rT^* v\text{ for some }v\in L}
  \end{align*}
\end{defi}
$\CT$ is indeed a closure operator, and $\CT(L)$ may alternatively be described as the set of guarded strings obtained by replacing occurrences of $\top$ in a guarded string of $L$ with arbitrary guarded strings (over $\Sigmat$).

For instance, we have
$\alpha a\alpha\top\beta\lT\alpha a\alpha b\gamma c\beta$ (by choosing $w=\alpha b\gamma c\beta$),
and
$\alpha\top\alpha\lT\alpha$ (by choosing $w=\alpha$).
Also observe that in a guarded string with shape $u\alpha\top\beta v$, we cannot replace $\top$ with a guarded string of length zero unless $\alpha=\beta$.

\begin{lem}
  \label{lem:ct:hom}
  $\CT$ is an $\St$-algebra homomorphism.
\end{lem}
\begin{proof}
  The only interesting case is that of composition, which amounts to showing that for all guarded strings $u,v,w$,
  \begin{align*}
    u \rT v\diamond w \text{ \quad iff\quad there are $v',w'$ s.t. }
    u=v'\diamond w' \text{ and }
    \begin{cases}
      \text{either $v'\rT v$ and $w'=w$,}\\
      \text{or $v'=v$ and $w'\rT w$}\\
    \end{cases}
  \end{align*}
  This follows from the fact that our rewriting relation $\rT$ replaces single letters.
\end{proof}

\begin{defi}[Expression closure $r$]
  \label{def:r}
  Let $r$ be the unique $S$-algebra homomorphism on expressions such that
  $r(a)=a$ for all letters $a\in\Sigma$, and $r(\top)=\fullexprt$ (where $\fullexprt$ is an expression  for the full guarded string language $\GS_\Sigmat$---e.g., $(a+b+\dots+\top)^*$).
\end{defi}

\begin{prop}
  \label{prop:r:props}
  For all expressions $e$, we have
  \begin{enumerate}[(i)]
  \item $[r(e)]=\CT[e]$, and
  \item $\KATT\proves e=r(e)$.
  \end{enumerate}
\end{prop}
\begin{proof}\hfill
  \begin{enumerate}[(i)]
  \item $[r(\cdot)]$ and $\CT[\cdot]$ are $S$-algebra homomorphisms, and they agree on $\Sigmat$.
  \item We proceed by induction on $e$; the only interesting case is when $e=\top$, for which we have $\KATT\proves r(\top)\leq\top$ by axiom~(\ref{ax:T}), and $\KATT\proves \top\leq r(\top)$ by completeness of \KAT~(\ref{eq:KAT'}), since $[\top]\subseteq\GS_\Sigmat=[r(\top)]$. \qedhere
  \end{enumerate}
\end{proof}

\begin{thm}
  \label{thm:lang}
  For all expressions $e,f$, we have
  \begin{align*}
  \GSL\models e=f \iff
  \CT[e]=\CT[f] \iff
  \KATT\proves e=f
  \end{align*}
\end{thm}
\begin{proof}
  We have
  \begin{align*}
    &\GSL\models e=f\\
    \tag{$\CT[\cdot]$ is an interpetation into a member of \GSL, by Lemma~\ref{lem:ct:hom}}
    \Rightarrow~ & \CT[e]=\CT[f]\\
    \tag{Proposition~\ref{prop:r:props}(i)}
    \Leftrightarrow~ & [r(e)]=[r(f)]\\
    \tag{completeness of \KAT~(\ref{eq:KAT})}
    \Leftrightarrow~ & \KAT\proves r(e)=r(f)\\
    \tag{transitivity and Proposition~\ref{prop:r:props}(ii)}
    \Rightarrow~ & \KATT\proves e=f\\
    \tag{soundness of \KATT axioms w.r.t. \GSL}
    \Rightarrow~ & \GSL\models e=f
  \end{align*}
  (In the last step, soundness w.r.t. \GSL comes from our assumption about \KA, and a trivial verification for axioms~\eqref{ax:atoms} and~\eqref{ax:T}.)
\end{proof}
The first equivalence in the statement of the above theorem could be obtained in a more direct way, without resorting to completeness of some axiomatisation.
Moreover the right-to-left implication of the second equivalence is an instance of a general property of closed language models~\cite[Theorem~2]{dkpp:fossacs19:kah}---duly generalised to the guarded string case. The reduction $r$ is used only for the left-to-right implication of this second equivalence.

According to the above proof, we could complete the statement with ``\dots $\iff [r(e)]=[r(f)]$''.
Doing so gives us a \pspace algorithm: compute the expressions $r(e)$ and $r(f)$, and compare them for guarded string language equivalence~\cite{cohenks96:kat:complexity}.
\begin{rem}
  \label{rem:bug}
  Note that it is crucial that $r(\top)$ be defined as an expression $\fullexprt$ for the full guarded string language over $\Sigmat$ rather than over $\Sigma$ alone. Otherwise, we would equate $\fullexpr$ and $\top$, while those are different in \GSL (e.g., for a counterexample when $\Sigma=\set{a,b}$, interpret both $a$ and $b$ as the empty language on some non-empty alphabet).
\end{rem}

\section{Relational models}
\label{sec:rel}

Given a set $X$, a \emph{relation on $X$} is a set of pairs of elements from $X$. We let $R,S$ range over such relations, whose set is written $\rel X$, and we write $\rin R x y$ for $\tuple{x,y}\in R$.
Given a function $p\colon X\to\Atom$,
relations on $X$ form an $\St$-algebra with the operations defined as follows:
\begin{align*}
  R+S &\eqdef R\cup S\\
  R\cdot S &\eqdef \set{\tuple{x,z}\mid \exists y\in X,~\rin R x y~\land~\rin S y z}\\
  R^* &\eqdef \set{\tuple{x_0,x_n}\mid \exists n\in\NN,x_1,\dots, x_{n-1},~\forall i<n,~\rin R {x_i} {x_{i+1}}}\\
  0 &\eqdef \emptyset\\
  1 &\eqdef \set{\tuple{x,x}\mid x\in X}\\
  \top &\eqdef X\times X\\
  \alpha &\eqdef \set{\tuple{x,x}\mid p(x)=\alpha}
\end{align*}
In words, $+$ is set-theoretic union, $\cdot$ is relational composition, $\cdot^*$ is reflexive transitive closure, $0$, $1$ and $\top$ are the empty, identity and full relations, respectively.
The function $p$ is only used to define the constants $\alpha$. The idea is that $p$ describes a partition of $X$, using atoms to name the equivalence classes. The relation $\alpha$ consists of the sub-identity relation selecting precisely the elements whose equivalence class is named $\alpha$.

We write \REL for the class of all $\St$-algebras of the above shape.

\begin{rem}
  Note that this definition covers the standard way of interpreting Kleene algebra with tests expressions into relations on a set $X$. Recall Remark~\ref{rem:std:kat}. If we start from a set $T$ of test variables, then an interpretation $v\colon T\to \pow(X)$ of test variables into predicates is the same as a function $p\colon X\to\Atom$ as above when setting $\Atom\eqdef 2^T$. Furthermore, the standard interpretation under $v$ of an atom $\alpha$ (seen as conjunction of literals in $T$) coincides with the one given in the above definition.

  In particular, while we set atoms as constants in our signatures, the ability to let $p$ vary in members of \REL gives them back their original status of variables.
\end{rem}

Let \KATF, \emph{KAT with a full element}, denote the union of the axioms from \KATT and axiom (\ref{ax:F}). Let us emphasise that despite the abbreviation, \KATF extends \KATT and thus contains axiom~(\ref{ax:T}). We prove in this section that \KATF is sound and complete for \REL.
The proof consists of two parts. First we characterise the equational theory of \REL in terms of closed guarded string languages (Section~\ref{ssec:rel:clang}, Proposition~\ref{prop:rel}), then we use reductions to show completeness of \KATF w.r.t.\ this closed language interpretation and obtain our main result (Section~\ref{ssec:rel:compl:clang}, Theorem~\ref{thm:rel}).

\subsection{Characterisation via closed guarded string languages}
\label{ssec:rel:clang}

From now on, all guarded strings and associated languages are over $\Sigmat$.
We start by extending the previous closure function (Definition~\ref{def:ct}), in order to take into account the new axiom (\ref{ax:F}):
\begin{defi}[Language closure $\CF$]
  \label{def:cf}
  Given two guarded strings $u,v$, we write $u\rF v$ if either $u\rT v$, or
  $u$ is obtained by replacing a subword of the shape $w\top w$ in $v$, with $w$ (for some guarded string $w$).
  Given a guarded string language $L$, we call \emph{$F$-closure of $L$} the guarded string language
  \begin{align*}
    \CF(L)\eqdef\set{u\mid u\rF^* v\text{ for some }v\in L}
  \end{align*}
\end{defi}
$\CF$ is a closure operator, but unlike $\CT$ in the previous section,
$\CF$ is not a homomorphism. For instance, $\CF(\set{\alpha a \alpha }\cdot \set{\alpha \top \alpha a \alpha })$ contains $\alpha a \alpha $ while $\CF(\set{\alpha a \alpha })\cdot\CF(\set{\alpha \top \alpha a\alpha })$ does not. Moreover, an elementary description of $\CF$ requires more work than for $\CT$ in the previous section.

Let $E$ be the following function on guarded string languages
\begin{align*}
  E(L) \eqdef \set{w \mid \exists n,~(w\top)^nw\in L}
\end{align*}
We shall prove that $\CF=E\circ \CT$, and that $\CF$ can be characterised in terms of certain graph homomorphisms (Proposition~\ref{prop:hom:cf} below).

\begin{defi}[Graph, graph homomorphism]
  \label{def:graph}
  A \emph{graph} is a tuple $\tuple{V,E,l,\iota,o}$, where $V$ is a set
  of \emph{vertices}, $E\subseteq V\times \Sigma \times V$ is a set of
  labelled edges, $l\colon V\to\Atom$ is a \emph{node-labelling function},
  and $\iota,o\in V$ are two distinguished vertices,
  respectively called \emph{input} and \emph{output}.

  A \emph{graph homomorphism} from the graph $G$ to the graph $H$ is a
  function from vertices of $G$ to vertices of $H$ that preserves node-labelling,
  labelled edges, input, and output. We write $H\lhd G$ when there
  exists a homomorphism from $G$ to $H$.
\end{defi}
The relation $\lhd$ on graphs is a preorder. We depict graphs as usual, using an unlabelled ingoing (resp. outgoing) arrow to indicate the input (resp. output); we use dotted red arrows to depict graph homomorphisms. For instance, we depict two finite connected graphs below, and a homomorphism between them:
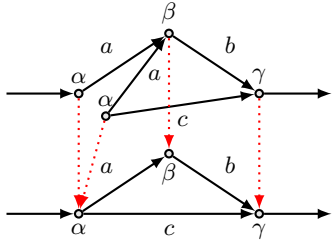
\begin{align*}
  \begin{tikzpicture}[xscale=1.2]
    \begin{scope}
      \vertexa 0 {0,0} \alpha;
      \vertexa 1 {1,.8} \beta;
      \vertexa 2 {2,0} \gamma;
      \vertexa {0'} {.3,-.3} \alpha;
      \inst 0; \fnst 2;
      \edge 0 1 a;
      \edger{0'}1 a;
      \edge 1 2 b;
      \edgeb{0'} 2 c;
    \end{scope}
    \begin{scope}[yshift=-1.6cm]
      \vertexb 3 {0,0} \alpha;
      \vertexb 4 {1,.8} \beta;
      \vertexb 5 {2,0} \gamma;
      \inst 3; \fnst 5;
      \edge 3 4 a;
      \edge 4 5 b;
      \edgeb 3 5 c;
    \end{scope}
    \homo 0  3;
    \homo{0'}3;
    \homo 1  4;
    \homo 2  5;
  \end{tikzpicture}
\end{align*}

\begin{defi}[Graph of a guarded string]
  \label{def:graph:word}
  We associate to each guarded string $u$ the graph $g(u)$ defined as follows:
  \begin{itemize}
  \item the vertices are the natural numbers smaller or equal to the length $n$ of $u$;
  \item the $i$th vertex is labelled with the $i$th atom occurring in $u$;
  \item for $a\in\Sigma$ there is an $a$-labelled edge from $i$ to $i+1$ if the $i$-th letter of $u$ is $a$;
  \item the input is $0$ and the output is $n$.
  \end{itemize}
\end{defi}

Graphs of guarded strings are rather simple: graphs as depicted above do not arise as graphs of guarded strings.
For guarded strings not containing $\top$, they are just directed paths from the input to the output. For guarded strings containing $\top$, they are collections of (possibly empty) directed paths where the input is the starting-point of some path and the output is the end-point of some path.
For example, the graphs of $\alpha a\beta b\gamma$ and $\alpha d\beta \top \alpha e\beta\top\gamma$ are depicted below:
\begin{align*}
  \begin{tikzpicture}[xscale=.8]
    \vertexa 0 {0,0} \alpha;
    \vertexa 1 {1,0} \beta;
    \vertexa 2 {2,0} \gamma;
    \inst 0; \fnst 2;
    \edge 0 1 a;
    \edge 1 2 b;
  \end{tikzpicture}
  &&
  \begin{tikzpicture}[xscale=.8]
    \vertexa 0 {0,0} \alpha;
    \vertexa 1 {1,0} \beta;
    \vertexa 2 {2,0} \alpha;
    \vertexa 3 {3,0} \beta;
    \vertexa 4 {4,0} \gamma;
    \inst 0; \fnst 4;
    \edge 0 1 d;
    \edge 2 3 e;
  \end{tikzpicture}
\end{align*}
Nevertheless, homomorphisms between graphs of guarded strings may be non-trivial. E.g., we have $g(\alpha a\beta b\gamma)\lhd g(\alpha a\beta\top\alpha a\beta b\gamma\top \beta b\gamma)$ and $g(\gamma\top \alpha a\alpha\top\beta b\beta\top\gamma)\lhd g(\gamma\top \beta b\beta\top \alpha a\alpha\top\gamma)$, as witnessed below:
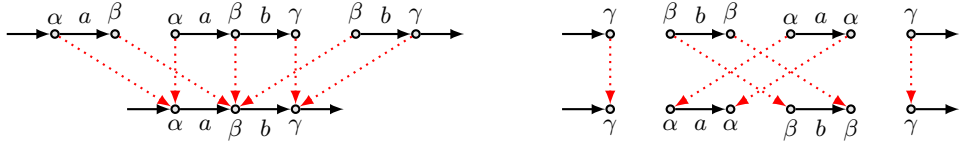
\begin{align*}
  \begin{tikzpicture}[xscale=.8]
    \vertexa{s0}{0,1}\alpha;
    \vertexa{s1}{1,1}\beta;
    \vertexa{s2}{2,1}\alpha;
    \vertexa{s3}{3,1}\beta;
    \vertexa{s4}{4,1}\gamma;
    \vertexa{s5}{5,1}\beta;
    \vertexa{s6}{6,1}\gamma;
    \inst{s0};\fnst{s6};
    \edge{s0}{s1}a;
    \edge{s2}{s3}a;
    \edge{s3}{s4}b;
    \edge{s5}{s6}b;
    \vertexb{t0}{2,0}\alpha;
    \vertexb{t1}{3,0}\beta;
    \vertexb{t2}{4,0}\gamma;
    \inst{t0}; \fnst{t2};
    \edgeb{t0}{t1}a;
    \edgeb{t1}{t2}b;
    \homo{s0}{t0};
    \homo{s1}{t1};
    \homo{s2}{t0};
    \homo{s3}{t1};
    \homo{s4}{t2};
    \homo{s5}{t1};
    \homo{s6}{t2};
  \end{tikzpicture}
  &&
  \begin{tikzpicture}[xscale=.8]
    \vertexa{s0}{0,1}\gamma;
    \vertexa{s1}{1,1}\beta;
    \vertexa{s2}{2,1}\beta;
    \vertexa{s3}{3,1}\alpha;
    \vertexa{s4}{4,1}\alpha;
    \vertexa{s5}{5,1}\gamma;
    \inst{s0};\fnst{s5};
    \edge{s1}{s2}b;
    \edge{s3}{s4}a;
    \vertexb{t0}{0,0}\gamma;
    \vertexb{t1}{1,0}\alpha;
    \vertexb{t2}{2,0}\alpha;
    \vertexb{t3}{3,0}\beta;
    \vertexb{t4}{4,0}\beta;
    \vertexb{t5}{5,0}\gamma;
    \inst{t0};\fnst{t5};
    \edgeb{t1}{t2}a;
    \edgeb{t3}{t4}b;
    \homo{s0}{t0};
    \homo{s1}{t3};
    \homo{s2}{t4};
    \homo{s3}{t1};
    \homo{s4}{t2};
    \homo{s5}{t5};
  \end{tikzpicture}
\end{align*}
In the sequel, we shall represent homomorphisms between graphs of guarded strings in a slightly more compact way, starting directly from the natural writing of the guarded strings, and using horizontal lines and shaded parallelograms to emphasise distinguished subwords and mappings between them. For instance, the above homomorphisms can be generalised to
$g(u\diamond v)\lhd g(u\top u\diamond v\top v)$ and $g(\alpha\top u\top v\top\beta)\lhd g(\alpha\top v\top u\top\beta)$ for arbitrary guarded strings $u,v$ and atoms $\alpha,\beta$, which we can represent as follows:
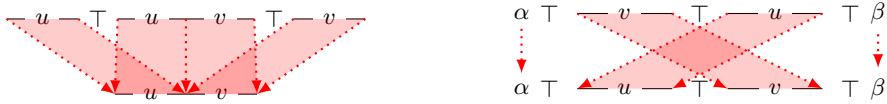
\begin{align*}
  \begin{tikzpicture}[xscale=.9]
    \node(s0) at (0.2,1) {};
    \node(s12) at (1.7,1) {$\top$};
    \node(s) at (3,1) {};
    \node(s45) at (4.3,1) {$\top$};
    \node(s6) at (5.8,1) {};
    \draw (s0) to node[fill=white]{$u$} (s12.west);
    \draw (s12.east) to node[fill=white]{$u$} (s.center);
    \draw (s.center) to node[fill=white]{$v$} (s45.west);
    \draw (s45.east) to node[fill=white]{$v$} (s6);
    \node(t0) at (2.1,0) {};
    \node(t) at (3,0) {};
    \node(t6) at (3.9,0) {};
    \draw (t0.west) to node[fill=white]{$u$} (t.center);
    \draw (t.center) to node[fill=white]{$v$} (t6.east);
    \shomo{s0.east}{t0.west}{s12.west}{t.center};
    \shomo{s12.east}{t0.west}{s45.west}{t6.east};
    \shomo{s45.east}{t.center}{s6.west}{t6.east};
    \homo{s.center}{t.center};
  \end{tikzpicture}
  &&
  \begin{tikzpicture}[xscale=2]
    \node(T1) at (0,1) {$~\top~$};
    \node(T2) at (1,1) {$~\top~$};
    \node(T3) at (2,1) {$~\top~$};
    \node(A) at (T1.west) {$\alpha$};
    \node(B) at (T3.east) {$\beta$};
    \draw (T1.east) to node[fill=white]{$v$} (T2.west);
    \draw (T2.east) to node[fill=white]{$u$} (T3.west);
    \node(T'1) at (0,0) {$~\top~$};
    \node(T'2) at (1,0) {$~\top~$};
    \node(T'3) at (2,0) {$~\top~$};
    \node(A') at (T'1.west) {$\alpha$};
    \node(B') at (T'3.east) {$\beta$};
    \draw (T'1.east) to node[fill=white]{$u$} (T'2.west);
    \draw (T'2.east) to node[fill=white]{$v$} (T'3.west);
    \homo{A}{A'};
    \shomo{T1.east}{T'2.east}{T2.west}{T'3.west};
    \shomo{T2.east}{T'1.east}{T3.west}{T'2.west};
    \homo{B}{B'};
  \end{tikzpicture}
\end{align*}
Our main interest in graphs and homomorphisms comes from the following characterisation of the equational theory of \REL.
Without atoms, this characterisation appeared first in~\cite[Theorem~6]{bp:lics15:paka}, for the syntax of Kleene allegories. Its (trivial) extension to Kleene allegories with top then appeared in~\cite[Theorem~16]{pous:stacs18}.
\begin{thm}
  \label{thm:rel:graphs}
  For all expressions $e,f$, we have:
  \begin{align*}
    \REL\models e\leq f \iff \forall u\in [e],\ \exists v\in [f],\ g(u)\lhd g(v)
  \end{align*}
\end{thm}
\begin{proof}
  Cf. above references. That we need the theorem only in a small fragment here (without intersection and converse) does not seem to enable substantial simplifications. In particular, we still need to consider arbitrary graphs, and a variant of~\cite[Lemma~3]{AB95} with top.
  That we deal with guarded string languages does not bring any difficulty, once we have the idea to label the graph vertices by atoms.
  We give a proof in Appendix~\ref{app:graphs} for the sake of completeness.
\end{proof}

\begin{rem}
  For guarded strings $u,v$ without top, we have $g(u)\lhd g(v)$ iff $u=v$.
  Therefore, for expressions $e,f$ without top (whose languages only contain guarded strings without top), the above theorem reduces to $\REL\models e\leq f \iff [e]\subseteq [f]$, a standard variant of one of the equivalences recalled in the introduction~\cite[Theorem~6]{kozens96:kat:completeness:decidability}.
\end{rem}

Thanks to Theorem~\ref{thm:rel:graphs}, it suffices to relate homomorphisms between graphs of guarded strings to the notion of $\CF$-closure. We do so in the following lemma.
\begin{lem}
  \label{lem:hom:cf}
  For all guarded strings $u,v$, the following are equivalent:
  \begin{enumerate}[(i)]
  \item\label{item:one} $u\rF^* v$,
  \item\label{item:three} $g(u)\lhd g(v)$,
  \item\label{item:two} $u\in E(\CT\set v)$.
  \end{enumerate}
\end{lem}
\begin{proof}
  We show \ref{item:one}~$\Rightarrow$~\ref{item:three}~$\Rightarrow$~\ref{item:two}~$\Rightarrow$~\ref{item:one}. For the first implication, since $\lhd$ is a preorder, it suffices to show that $u\rF v$ entails $g(u)\lhd g(v)$. There are two cases to consider.
  \begin{itemize}
  \item either the rewriting rule associated to axiom (\ref{ax:T}) was used, i.e., $u=l\diamond w\diamond r$ and $v=l\top r$ for some guarded strings $l,w,r$. In that case we have the following homomorphism from the graph of $v$ to the graph of $u$:
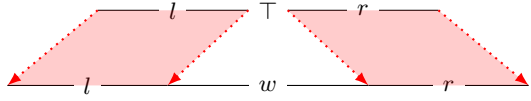
    \begin{align*}
      \begin{tikzpicture}[xscale=1.2]
        \node (l) at (0,1) {};
        \node (T) at (2,1) {$\top$};
        \node (r) at (4,1) {};
        \draw (l) to node[fill=white]{$l$} (T.west);
        \draw (T.east) to node[fill=white]{$r$} (r);
        \node (l') at (-1,0) {};
        \node (wl) at (1,0) {};
        \node (wr) at (3,0) {};
        \node (r') at (5,0) {};
        \draw (l') to node[fill=white]{$l$} (wl);
        \draw (wl.west) to node[fill=white]{$w$} (wr.east);
        \draw (wr) to node[fill=white]{$r$} (r');
        \shomo{l.east}{l'.east}{T.west}{wl.west};
        \shomo{T.east}{wr.east}{r.west}{r'.west};
      \end{tikzpicture}
    \end{align*}
  \item or the rewriting rule associated to axiom (\ref{ax:F}) was used, i.e., $u=lwr$ and $v=lw\top wr$ for some words $l,r$ and guarded string $w$. In that case we have the following homomorphism from the graph of $v$ to the graph of $u$:
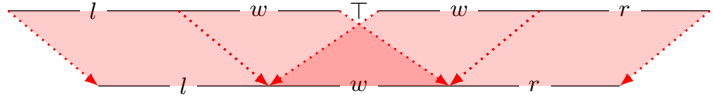
    \begin{align*}
      \begin{tikzpicture}[xscale=1.2]
        \node (l) at (-2,1) {};
        \coordinate (u) at (0,1);
        \node (T) at (2,1) {$\top$};
        \coordinate (v) at (4,1);
        \node (r) at (6,1) {};
        \draw (l) to node[fill=white]{$l$} (u.west);
        \draw (u.east) to node[fill=white]{$w$} (T.west);
        \draw (T.east) to node[fill=white]{$w$} (v.west);
        \draw (v.east) to node[fill=white]{$r$} (r);
        \node (l') at (-1,0) {};
        \coordinate (wl) at (1,0);
        \coordinate (wr) at (3,0);
        \node (r') at (5,0) {};
        \draw (l') to node[fill=white]{$l$} (wl);
        \draw (wl) to node[fill=white]{$w$} (wr);
        \draw (wr) to node[fill=white]{$r$} (r');
        \shomo{l.east}{l'.east}{u.west}{wl.west};
        \shomo{u.east}{wl.east}{T.west}{wr.west};
        \shomo{T.east}{wl.east}{v.west}{wr.west};
        \shomo{v.east}{wr.east}{r.west}{r'.west};
      \end{tikzpicture}
    \end{align*}
  \end{itemize}
  For the second implication, assume $g(u)\lhd g(v)$. Let $n$ be the number of occurrences of $\top$ in $v$, and let $v_0,\dots,v_n$ be the top-free guarded strings such that $v=v_0\top v_1\top \cdots \top v_n$.
  Since they are top-free, those subwords must be mapped to subwords of $u$. For instance, when $n=3$, the homomorphism may look as follows:
  
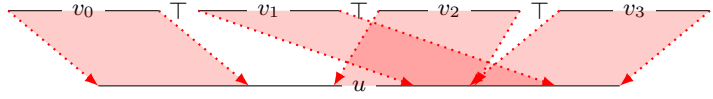
\begin{align*}
    \begin{tikzpicture}[xscale=1.2]
      \node (v) at (0,1) {};
      \node (v') at (8,1) {};
      \node (T1) at (2,1) {$\top$};
      \node (T2) at (4,1) {$\top$};
      \node (T3) at (6,1) {$\top$};
      \draw (v) to node[fill=white]{$v_0$} (T1.west);
      \draw (T1.east) to node[fill=white]{$v_1$} (T2.west);
      \draw (T2.east) to node[fill=white]{$v_2$} (T3.west);
      \draw (T3.east) to node[fill=white]{$v_3$} (v');
      \node (u) at (1,0) {};
      \node (u') at (7,0) {};
      \draw (u) to node[fill=white]{$u$} (u');
      \shomo{v.east}{u.east}{T1.west}{$(T1.west)-(-1,1)$};
      % \shomo{T1.east}{$(T2.east)-(.8,1)$}{T2.west}{$(T3.west)-(.8,1)$};
      % \shomo{T2.east}{$(T1.east)-(.2,1)$}{T3.west}{$(T2.west)-(.2,1)$};
      \shomo{T1.east}{$(T1.east)-(-2.4,1)$}{T2.west}{$(T2.west)-(-2.4,1)$};
      \shomo{T2.east}{$(T2.east)-(.5,1)$}{T3.west}{$(T3.west)-(.5,1)$};
      \shomo{T3.east}{$(T3.east)-(1,1)$}{v'.west}{u'.west};
    \end{tikzpicture}
  \end{align*}
  For all $0\leq i\leq n$, let $l_i,r_i$ be the words such that $u=l_iv_ir_i$. We have that $l_0$ and $r_n$ must be the empty word since inputs and outputs must be preserved by homomorphisms.
  We have $(u\top)^nu \rT^n v$: we can obtain $(u\top)^nu$ from $v$ by replacing the $i$th occurrence of $\top$ in $v$ with the word $r_{i-1}\top l_i$, for $0<i\leq n$.
  This suffices to conclude that $u\in E(\CT\set v)$: we have proven \ref{item:three}~$\Rightarrow$~\ref{item:two}.
  As an example, when $n=3$, the situation may be depicted as follows:
  
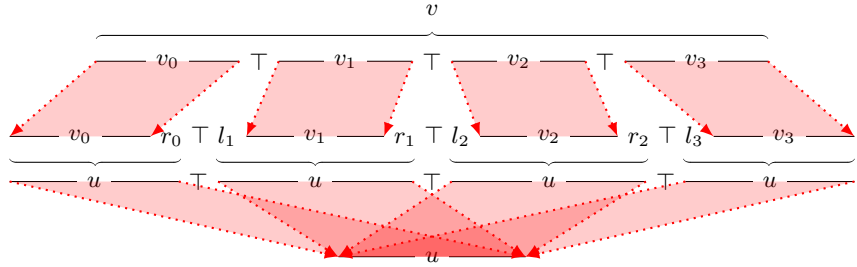
\begin{align*}
    \begin{tikzpicture}[xscale=1.15]
      \node (v) at (0,2) {};
      \node (T1) at (2,2) {$\top$};
      \node (T2) at (4,2) {$\top$};
      \node (T3) at (6,2) {$\top$};
      \node (v') at (8,2) {};
      \draw (v) to node[fill=white]{$v_0$} (T1.west);
      \draw (T1.east) to node[fill=white]{$v_1$} (T2.west);
      \draw (T2.east) to node[fill=white]{$v_2$} (T3.west);
      \draw (T3.east) to node[fill=white]{$v_3$} (v');
      \node (l) at (-1,1) {};
      \node (T'1) at (1.3,1) {$r_0~\top~l_1$};
      \node (T'2) at (4,1) {$r_1~\top~l_2$};
      \node (T'3) at (6.7,1) {$r_2~\top~l_3$};
      \node (r) at (9,1) {};
      \draw (l) to node[fill=white]{$v_0$} (T'1.west);
      \draw (T'1.east) to node[fill=white]{$v_1$} (T'2.west);
      \draw (T'2.east) to node[fill=white]{$v_2$} (T'3.west);
      \draw (T'3.east) to node[fill=white]{$v_3$} (r);
      \shomo{v.east}{l.east}{T1.west}{T'1.west};
      \shomo{T1.east}{T'1.east}{T2.west}{T'2.west};
      \shomo{T2.east}{T'2.east}{T3.west}{T'3.west};
      \shomo{T3.east}{T'3.east}{v'.west}{r.west};
      \draw [decorate,decoration={calligraphic brace}]($(v.east)+(0,.3)$)--($(v'.west)+(0,.3)$)
        node[pos=.5, above=5pt]{$v$};
      \draw [decorate,decoration={calligraphic brace}]($(T'1.west)+(.35,-.3)$)--($(l.east)-(0,.3)$);
      \draw [decorate,decoration={calligraphic brace}]($(T'2.west)+(.35,-.3)$)--($(T'1.east)-(.35,.3)$);
      \draw [decorate,decoration={calligraphic brace}]($(T'3.west)+(.35,-.3)$)--($(T'2.east)-(.35,.3)$);
      \draw [decorate,decoration={calligraphic brace}]($(r.west)+(0,-.3)$)--($(T'3.east)-(.35,.3)$);
      \begin{scope}[yshift=4mm]
        \node (l') at (-1,0) {};
        \node (T''1) at (1.3,0) {$\top$};
        \node (T''2) at (4,0) {$\top$};
        \node (T''3) at (6.7,0) {$\top$};
        \node (r') at (9,0) {};
        \draw (l') to node[fill=white]{$u$} (T''1.west);
        \draw (T''1.east) to node[fill=white]{$u$} (T''2.west);
        \draw (T''2.east) to node[fill=white]{$u$} (T''3.west);
        \draw (T''3.east) to node[fill=white]{$u$} (r');
        \node (u) at (2.8,-1) {};
        \node (u') at (5.2,-1) {};
        \draw (u) to node[fill=white]{$u$} (u');
        \shomo{l'.east}{u.east}{T''1.west}{u'.west};
        \shomo{T''1.east}{u.east}{T''2.west}{u'.west};
        \shomo{T''2.east}{u.east}{T''3.west}{u'.west};
        \shomo{T''3.east}{u.east}{r'.west}{u'.west};
      \end{scope}
    \end{tikzpicture}
  \end{align*}
  For the last implication, assume that $u\in E(\CT\set v)$. There exists $n$ such that $(u\top)^nu\rT^*v$, and thus in particular $(u\top)^nu\rF^*v$. Finally observe that $u\rF^n(u\top)^nu$ using $n$ rewriting steps using (\ref{ax:F}), so that we can conclude by transitivity: $u\rF^n(u\top)^nu\rF^*v$.
\end{proof}

The above lemma has two important immediate consequences.
First we have the announced factorisation of the closure $\CF$, and second, combined with Theorem~\ref{thm:rel:graphs}, we obtain a characterisation of the equational theory of \REL in terms of closed languages:
\begin{prop}
  \label{prop:hom:cf}
  We have $\CF=E\circ \CT$.
\end{prop}
\begin{prop}
  \label{prop:rel}
  For all expressions $e,f$, we have:
  \begin{align*}
    \REL\models e=f \iff \CF[e]=\CF[f]
  \end{align*}
\end{prop}
\begin{proof}
  For all $e,f$, we have:
  \begin{align*}
    \tag{by Theorem~\ref{thm:rel:graphs}}
    \REL\models e\leq f & \iff \forall u\in [e],\ \exists v\in [f],\ g(u)\lhd g(v) \\
    \tag{by Lemma~\ref{lem:hom:cf}}
    & \iff [e]\subseteq\CF[f]
  \end{align*}
  The initial statement follows by antisymmetry and the fact that $\CF$ is a closure (so that for all languages $L,K$, $L\subseteq \CF(K)$ iff $\CF(L)\subseteq \CF(K)$).
\end{proof}

\subsection{Completeness w.r.t.\ closed guarded string languages}
\label{ssec:rel:compl:clang}

It remains to show that \KATF is complete w.r.t.\ the previous closed language interpretation $(\CF[\cdot])$.
We use reductions in order to do so: we find a counterpart to the function $r$ from Section~\ref{sec:lang} (Definition~\ref{def:r}), for the $F$-closure rather than the $T$-closure.
By Proposition~\ref{prop:hom:cf}, and since we already have the function $r$ for $T$-closure, it actually suffices to find a function $s$ that corresponds to the function $E$, i.e., such that for all expressions $e$, $s(e)$ is an expression whose language is $E[e]$ and such that $\KATF\proves e=s(e)$.

\subsubsection{Interlude: monoids and the square root of a language}
\label{sssec:interlude}

Before defining the aforementioned reduction $s$ for $E$, let us consider an exercise about plain regular languages.
Define the \emph{square root of a language $L$} as follows:
\begin{align*}
  \sqrt{L}\eqdef\set{w\mid w^2 \in L}
\end{align*}
How to prove that this operation preserves regularity (i.e., if $L$ is regular, then so is $\sqrt{L}$)?
This is not obvious with automata-based techniques, but there is a simple and elegant solution using the characterisation of regular languages as those recognised by finite monoids~\cite{eilenberg1974automata,sakarovitch_2009}.

Let us recall the corresponding definitions.
A \emph{monoid} is a tuple $\tuple{M,\cdot,1}$ where $M$ is a set and $\cdot$ is an associative binary operation on $M$ with $1$ as neutral element; it is \emph{finite} when $M$ is so.
Words on $\Sigma$ form a monoid, in fact the free monoid over $\Sigma$.
A language $L$ is \emph{recognised by a finite monoid} $M$ if there exists a homomorphism $h$ from words to $M$ and a subset $P$ of $M$ such that $L=\ih(P)$ (i.e., for all words $w$, $w\in L$ iff $h(w)\in P$).

\begin{prop}
  \label{prop:reg:rec}
  A language is regular iff it is recognised by a finite monoid.
\end{prop}
This result is entirely standard; we recall a proof below which will be helpful later to understand our construction on guarded string languages, and its complexity.
\begin{proof}\hfill
  \begin{itemize}
  \item Given a finite non-deterministic automaton for a language $L$, the binary relations on its state-space form a finite monoid.
    Consider the unique homomorphism $h$ mapping a letter to its transition relation; this homomorphism maps a word $u$ to the relation containing all pairs of states that can be related by a $u$-labelled path in the automaton. Let $P$ consist of all relations containing at least one pair of an initial state and a final state. We have $L=\ih(P)$ by construction.
  \item Given a finite monoid $\tuple{M,\cdot,1}$, a homomorphism $h$ and a subset $P$ such that $L=\ih(P)$, we construct a finite deterministic automaton for $L$ as follows: states are elements of the monoid; the initial state is the neutral element $1$; the transition function maps a state $x$ and a letter $a$ to the state $x\cdot h(a)$; accepting states are the elements of $P$.
    \qedhere
  \end{itemize}
\end{proof}
Let us also recall the following basic property:
\begin{lem}
  \label{lem:monoid}
  Let $x,y$ be elements of a monoid and let $h$ be a homomorphism from words to that monoid.
  We have the following language inclusion:
  \begin{align*}
    \ih(x)\cdot \ih(y) \subseteq \ih(x\cdot y)
  \end{align*}
\end{lem}
\begin{proof}
  A word in the left-hand side has shape $uv$ for some words $u,v$ such that $h(u)=x$ and $h(v)=y$.
  We then check that $h(uv)=h(u)\cdot h(v)=x\cdot y$, as required.
\end{proof}

To solve the exercise, let us generalise the previous square root function to subsets of arbitrary monoids.
Given a subset $P$ of a monoid, let $\sqrt P$ be the following subset:
\begin{align*}
  \sqrt{P}\eqdef\set{x\mid x^2 \in P}
\end{align*}
\begin{prop}
  \label{prop:square}
  If $L=\ih(P)$, then $\sqrt L=\ih(\sqrt P)$
\end{prop}
\begin{proof}
  For all words $w$, we have
  \begin{align*}
    w\in\sqrt L
    \iff~& w^2 \in L\\
    \iff~& h(w^2) \in P\\
    \tag{$h$ is a homomorphism}
    \iff~& h(w)^2 \in P\\
    \iff~& h(w) \in \sqrt P\\
    \tag*\qedhere
    \iff~& w \in \ih(\sqrt P)
  \end{align*}
\end{proof}
Thus, if $L$ is regular, so is $\sqrt L$: we have solved our exercise.
In particular, we can obtain a square root function on regular expressions such that $\plang{\sqrt e}=\sqrt{\plang e}$ for all $e$.

\medskip

Adapting the above idea, we will obtain a function $s$ such that $[s(e)]=E[e]$ for all expressions $e$.
This does not explain how to show $\KATF\proves s(e)\leq e$, however. We first show how to do so with our square root example. There, the counterpart to axiom~\eqref{ax:F} is the inequation $x\leq x^2~(S)$. Accordingly, let \KAS denote the union of \KA and $(S)$.
\begin{prop}
  For all regular expressions $e$, we have $\KAS\proves \sqrt e \leq e$.
\end{prop}
\begin{proof}
  Since $\plang e$ is regular,
  we have $\plang e=\ih(P)=\bigcup_{p\in P}\ih(p)$ for some homomorphism $h$ and subset $P$ of a finite monoid.
  For all elements $x$ in that monoid, let $e_x$ be a regular expression for the language $\ih(x)$ (this language being regular by Proposition~\ref{prop:reg:rec}).

  The previous language equation can be rewritten as $\plang e=\plang{\sum_{p\in P}e_p}$.
  Similarly, we have $\plang{\sqrt e}=\sqrt{\plang e}=\plang{\sum_{p\in \sqrt P}e_p}=\plang{\sum_{p^2\in P}e_p}$.
  By completeness of \KA and then using axiom~$(S)$, we deduce
  \begin{align*}
    \KAS\proves \sqrt e = \sum_{p^2\in P} e_p\leq \sum_{p^2\in P} e_p^2
  \end{align*}
  Therefore, to prove the announced statement, it suffices to show that $\KAS\proves e_p^2\leq e$ whenever $p^2\in P$. This follows by completeness of \KA, since we have
  \begin{align*}
    \plang{e_p^2}
    =\plang{e_p}^2
    =\ih(p)^2
    \subseteq \ih(p^2)
    =\plang{e_{p^2}}
    \subseteq\plang e
  \end{align*}
  (Where we use Lemma~\ref{lem:monoid} for the first inclusion, and $p^2\in P$ for the second one.)
\end{proof}

\subsubsection{Monoids for guarded string languages}
\label{sssec:gmonoids}

On regular expressions with top but without atoms, one can easily adapt the previous idea to obtain a reduction for $E$.
Indeed, if $L=\ih(P)$, then $E(L)=\ih(P')$ for $P'$ defined as follows:
\begin{align*}
  P' = \set{x\mid \exists n, (x\cdot h(\top))^n\cdot x \in P}
\end{align*}
(Note that $P'$ can be computed when the monoid is finite: there are only finitely many powers of $x\cdot h(\top)$.)
This approach gives a simpler alternative to the automata construction we defined in~\cite[Section~4.2]{pw:concur22:katop}.

However, to deal with full expressions and guarded string languages, we first need to extend the theory of recognition by finite monoids to such languages.

To do so, we rely on the presentation of guarded strings as elements of $(\Atom\times\Sigmat)^*\times \Atom$. A word over $\Atom\times\Sigmat$ can be seen as a word over $\SigmatA$ (e.g., $(\alpha,a)(\beta,b)=\alpha a\beta b$), and when $u$ is such a word and $\alpha$ is an atom, $u\alpha$ is a guarded string. All guarded strings can be decomposed in this way (uniquely).

Every regular expression $\e$ over $\Atom\times\Sigmat$ can be seen as an expression $\underline\e$ by replacing each letter $(\alpha,a)$ by the expression $\alpha\cdot a$. We call \emph{clean} the  expressions of the form $\underline\e$. The languages of such expressions are almost guarded string languages: their words only miss the final atom.
\begin{fact}
  \label{fact:clean}
  For all clean expressions $e$ and atoms $\alpha$, we have $[e\cdot \alpha]=\plang{e}\cdot\set\alpha$.
\end{fact}

\begin{defi}
  \label{def:reco}
  A \emph{recogniser} is a tuple $\tuple{M,h,P}$ where $M$ is a monoid, $h$ is a homomorphism from $(\Atom\times\Sigmat)^*$ to $M$, and $P$ is a subset of $M\times\Atom$;
  it is \emph{finite} when $M$ is so.
  Given such a recogniser, we define the following guarded string language:
  \begin{align*}
    \ih(P)\eqdef \set{u\alpha \mid P(h(u),\alpha)}
  \end{align*}
\end{defi}
We also keep the notation from Section~\ref{sssec:interlude}: when $x$ is an element of the monoid, $\ih(x)=\set{u\mid h(u)=x}$ is a language over $\Atom\times\Sigmat$. In particular, we have
\begin{align*}
  \ih(P)=\bigcup_{P(p,\alpha)}\ih(p)\cdot\set\alpha
\end{align*}

\begin{prop}
  \label{prop:kleene1}
  For all expressions $e$, there exists a finite recogniser $\tuple{M,h,P}$ such that $[e]=\ih(P)$.
\end{prop}
\begin{proof}
  First we compute a non-deterministic finite automaton over the alphabet $\SigmatA$, for $\plang e$. Such an automaton is a tuple $\tuple{S,I,\Delta,F}$ where $S$ is a finite set of states, $I,F\subseteq S$ are the initial and accepting states, respectively, and $\Delta$ is the transition relation, seen as a map from $\SigmatA$ to relations on $S$. By construction, we have
  \begin{align*}
    x_1\dots x_n \in \plang e \iff \Delta(x_1)\cdot\dots\cdot \Delta(x_n) \cap I{\times} F \neq \emptyset
  \end{align*}
  (For all words $x_1\dots x_n$ over $\SigmatA$.)
  Recall the function $\gs$ extracting the guarded strings of a language, such that $[e]=\gs\plang e$.
  We deduce that for all guarded strings $\alpha_0a_1\dots a_{n-1}\alpha_n$, we have
  \begin{align*}
    \alpha_0a_1\dots a_{n-1}\alpha_n \in [e] \iff \Delta(\alpha_0)^*\cdot\Delta(a_1)\cdot\dots\cdot\Delta(a_{n-1})\cdot \Delta(\alpha_n)^* \cap I{\times} F \neq \emptyset
  \end{align*}
  Accordingly, we construct a finite recogniser as follows: $M$ is the monoid of relations on $S$; $h$ is the unique homomorphism such that $h(\alpha,a)=\Delta(\alpha)^*\cdot \Delta(a)$ for all atoms $\alpha$ and $a\in\Sigmat$; and $P(R,\alpha)$ holds if $R\cdot \Delta(\alpha)^* \cap I{\times} F\neq \emptyset$.
\end{proof}

The converse also holds: one can associate an expression to every finite recogniser. We need a slightly stronger statement in the sequel:
\begin{prop}
  \label{prop:kleene2}
  For all finite recognisers $\tuple{M,h,P}$, there are clean expressions $(e_x)_{x\in M}$ such that for all $x\in M$, $\plang{e_x} = \ih(x)$. It follows that
  \begin{align*}
    \ih(P) = \biggl[\sum_{P(p,\alpha)} e_p\cdot\alpha\biggr]
  \end{align*}
\end{prop}
\begin{proof}
  For each $x\in M$, let $\e_x$ be a regular expression over $\Atom\times\Sigmat$ for $\ih(x)$ and set $e_x\eqdef \underline{\e_x}$.
  We deduce via Fact~\ref{fact:clean} that
  \begin{align*}
    \tag*\qedhere
    \ih(P)
    = \bigcup_{P(p,\alpha)}\ih(p)\cdot\set\alpha
    = \bigcup_{P(p,\alpha)}\plang{e_p}\cdot\set\alpha
    = \bigcup_{P(p,\alpha)}[e_p\cdot \alpha]
    = \biggl[\sum_{P(p,\alpha)} e_p\cdot\alpha\biggr]
  \end{align*}
\end{proof}

\subsubsection{Reduction for $E$ and completeness}
\label{sssec:E:red}

Now that the monoid machinery is set up for guarded string languages, we can show that the function $E$ preserves regularity. % (of guarded string languages).
\begin{prop}
  \label{prop:E}
  Let $\tuple{M,h,P}$ be a recogniser and define $P'$ as follows:
  \begin{align*}
    P'\eqdef \set{(x,\alpha) \mid \exists n\in\NN,~P((x\cdot h(\alpha,\top))^n\cdot x,\alpha)}
  \end{align*}
  We have $E(\ih(P))=\ih(P')$.
\end{prop}
\begin{proof}
  For all guarded strings $u\alpha$, we have
  \begin{align*}
    u\alpha \in E(\ih(P))
    \iff~& \exists n\in\NN,~(u\alpha\top)^nu\alpha \in\ih(P)\\
    \iff~& \exists n\in\NN,~P(h((u\alpha\top)^nu),\alpha)\\
    \tag{$h$ is a homomorphism}
    \iff~& \exists n\in\NN,~P((h(u)\cdot h(\alpha,\top))^n\cdot h(u),\alpha)\\
    \iff~& P'(h(u),\alpha)\\
    \tag*\qedhere
    \iff~& u\alpha \in \ih(P')
  \end{align*}
\end{proof}

\begin{fact}
  \label{fact:P'}
  In the above proposition, when the recogniser is obtained from a transition monoid (Proposition~\ref{prop:kleene1}), we have the following alternative presentation of $P'$:
  \begin{align*}
    P'(X,\alpha) = P((X\cdot T_\alpha)^*\cdot X,\alpha) \quad\text{where}\quad T_\alpha\eqdef h(\alpha,\top)=\Delta(\alpha)^*\cdot\Delta(\top)
  \end{align*}
\end{fact}
\begin{proof}
  Because $\exists n,~Y^n\cdot Z\cap I{\times}F\neq\emptyset$
  iff $(\bigcup_nY^n\cdot Z)\cap I{\times}F\neq\emptyset$ iff $(Y^*\cdot Z)\cap I{\times}F\neq\emptyset$.
\end{proof}

We can finally define the reduction $s$. Given an expression $e$, first compute a finite recogniser $\tuple{M,h,P}$ such that $[e]=\ih(P)$ (Proposition~\ref{prop:kleene1}), then update $P$ into $P'$ as in Proposition~\ref{prop:E}, and finally extract the expression $s(e)$ from $\tuple{M,h,P'}$ (Proposition~\ref{prop:kleene2}).

\begin{prop}
  \label{prop:s:props}
  For all expressions $e$, we have
  \begin{enumerate}[(i)]
  \item $[s(e)]=E[e]$, and
  \item $\KATF\proves e=s(e)$.
  \end{enumerate}
\end{prop}
\begin{proof}
  The first item follows by construction and the three previous propositions.
  For the second one, since $[e]\subseteq E[e]=[s(e)]$, we have $\KAT\proves e\leq s(e)$ by completeness of \KAT~\eqref{eq:KAT'}, so that it suffices to show the other inequation.

  Let $\tuple{M,h,P}$ be the finite recogniser for $[e]$ used to construct $s(e)$ and
  let $(e_x)_{x\in M}$ be the expressions given by Proposition~\ref{prop:kleene2}.

  We have $[e]=[\sum_{P(p,\alpha)}e_p\cdot \alpha]$ and $[s(e)]=[\sum_{P'(p,\alpha)}e_p\cdot \alpha]$.
  By completeness of \KAT~\eqref{eq:KAT}, we deduce
  \begin{align*}
    \KAT\proves s(e)=\sum_{P'(p,\alpha)}e_p\cdot \alpha
  \end{align*}
  Therefore, it suffices to show that $\KATF\proves e_p\cdot \alpha\leq e$ whenever  $P'(p,\alpha)$. Accordingly, let $n,p,\alpha$ be such that $P((p\cdot h(\alpha,\top))^n\cdot p,\alpha)$. Set $t\eqdef h(\alpha,\top)$ and $q\eqdef (p\cdot t)^n\cdot p$; we have $P(q,\alpha)$.
  We derive
  \begin{align*}
    \tag{using axiom~\eqref{ax:F} $n$ times}
    \KATF\proves e_p\cdot \alpha
    &\leq (e_p\cdot \alpha \cdot \top)^n\cdot e_p\cdot \alpha\\
    &\leq e_q\cdot \alpha \\
    &\leq e
  \end{align*}
  For the last two steps, we use \KAT completeness~\eqref{eq:KAT'}: we have
  \begin{align*}
    [(e_p\cdot \alpha \cdot \top)^n\cdot e_p\cdot \alpha]
    \tag{Fact~\ref{fact:clean}: $(e_p\cdot \alpha \cdot \top)^n\cdot e_p$ is clean}
    &= \plang{(e_p\cdot \alpha \cdot \top)^n\cdot e_p}\cdot\set\alpha\\
    \tag{$\plang\cdot$ is a homomorphism}
    &= (\plang{e_p}\cdot \set{\alpha\top})^n\cdot \plang{e_p}\cdot\set\alpha\\
    \tag{definition of $e_p$---Proposition~\ref{prop:kleene2}}
    &= (\ih(p)\cdot \set{\alpha\top})^n\cdot \ih(p)\cdot\set\alpha\\
    \tag{$h(\alpha,\top)=t$}
    &\subseteq (\ih(p)\cdot \ih(t))^n\cdot \ih(p)\cdot\set\alpha\\
    \tag{by Lemma~\ref{lem:monoid}}
    &\subseteq \ih((p\cdot t)^n\cdot p)\cdot\set\alpha\\
    \tag{definition of $q$}
    &= \ih(q)\cdot\set\alpha\\
    \tag{definition of $e_q$---Proposition~\ref{prop:kleene2}}
    &= \plang{e_q}\cdot\set\alpha\\
    \tag{Fact~\ref{fact:clean}}
    &= [e_q\cdot\alpha]\\
    \tag*{($P(q,\alpha)$)~\qedhere}
    &\subseteq [e]
  \end{align*}
\end{proof}
We finally combine all the above results to obtain our main theorem:
\begin{thm}
  \label{thm:rel}
  For all regular expressions with top $e,f$, we have
  \begin{align*}
  \REL\models e=f \iff
  \CF[e]=\CF[f] \iff
  \KATF\proves e=f
  \end{align*}
\end{thm}
\begin{proof}
  We have
  \begin{align*}
    &\REL\models e=f\\
    \tag{Proposition~\ref{prop:rel}}
    \Leftrightarrow~ & \CF[e]=\CF[f]\\
    \tag{by Proposition~\ref{prop:hom:cf}}
    \Leftrightarrow~ & E(\CT[e])=E(\CT[f])\\
    \tag{by Propositions~\ref{prop:r:props}(i) and~\ref{prop:s:props}(i)}
    \Leftrightarrow~ & [s(r(e))]=[s(r(f))]\\
    \tag{by completeness of \KAT~(\ref{eq:KAT})}
    \Leftrightarrow~ & \KAT\proves s(r(e))=s(r(f))\\
    \tag{by transitivity and Propositions~\ref{prop:r:props}(ii) and~\ref{prop:s:props}(ii)}
    \Rightarrow~ & \KATF\proves e=f\\
    \tag*{(soundness of \KATF axioms w.r.t. \REL)~\qedhere}
    \Rightarrow~ & \REL\models e=f
  \end{align*}
  % (In the last step, soundness w.r.t. \REL comes from our assumption about \KA, and a simple verification for axioms~\eqref{ax:atoms}, ~(\ref{ax:T}) and~(\ref{ax:F}).)
\end{proof}
The above proof follows the same strategy as the one for Theorem~\ref{thm:lang}.
Like there, the right-to-left implication of the second equivalence in the statement is an instance of~\cite[Theorem~2]{dkpp:fossacs19:kah} (generalised to guarded string languages), and we use reductions only for the left-to-right part of this equivalence.

\subsection{Solution to the exercise from the introduction}
\label{ssec:ex}

Recall the exercise~\eqref{eq:ex} from the introduction.
Since this example does not involve tests, we work with \KA rather than \KAT, we identify words and guarded strings, and we simplify definitions according to Remark~\ref{rem:singleatom}: there is only one atom.
We first give a handcrafted solution, before illustrating how the previous construction works on this example.

\subsubsection{Handcrafted solution}
For a number $i\geq 2$, let us write $i$ for the expression $a^i$.
We have to prove the following inequation using \KA and axiom~\eqref{ax:F}:
\begin{equation*}
  3^* \leq 3^*\top 2^* + 2^*a\top 3^*
\end{equation*}
As a first hint, observe that $\KA\proves i^*\leq j^*$ when $i$ is a multiple of $j$.\\
As a second hint, let us prove:
\begin{equation*}
  6^* \leq 3^*\top 2^*
\end{equation*}
Indeed, we have $6^*\leq 6^*\top 6^*$ by axiom~\eqref{ax:F}, and both $6^*\leq 3^*$ and $6^*\leq 2^*$ since $6$ is a multiple of both $3$ and $2$.\\
Similarly, we can also prove:
\begin{equation*}
  6^*3 \leq 2^*a\top 3^*
\end{equation*}
Indeed, we have $6^*3\leq 6^*3\top 6^*3$ by axiom~\eqref{ax:F}, and then $6^*3\leq 2^*2a \leq 2^*a$ and $6^*3\leq 3^*3\leq 3^*$ by basic \KA reasoning.

Finally observe that $[3^*]\subseteq [6^*+6^*3]$: every multiple of $3$ is also either even or odd. This suffices to conclude by \KA completeness:
\begin{align*}
  3^* \leq 6^*+6^*3 \leq 3^*\top 2^* + 2^*a\top 3^*
\end{align*}

\subsubsection{Computed solution}
Now let us see how this solution can be obtained via our generic construction.
We purposely skip the reduction $r$ for axiom~\eqref{ax:T}, which we want to avoid.
Set $e\eqdef 3^*\top 2^* + 2^*a\top 3^*$.
We will get $\KA_F\proves 3^* \leq s(e) = e$ using \KA completeness for the first step (after checking that $[3^*]\subseteq E[e] = [s(e)]$), and Proposition~\ref{prop:s:props}(ii) for the second step.

In order to compute $s(e)$, consider the following automaton for $[e]$:
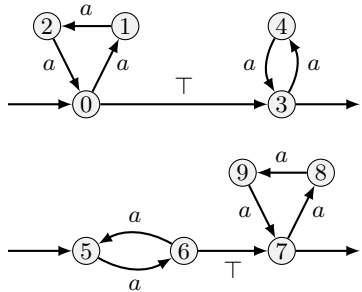
\begin{align*}
  \begin{tikzpicture}[scale=1.3]
    \node[state](0) at (0,2) {0};
    \node[state](1) at (0.4,2.8) {1};
    \node[state](2) at (-0.4,2.8) {2};
    \node[state](3) at (2,2) {3};
    \node[state](4) at (2,2.8) {4};
    \inst 0; \fnst 3;
    \draw[arc] (0) to node[right]{$a$} (1);
    \draw[arc] (1) to node[above]{$a$} (2);
    \draw[arc] (2) to node[left]{$a$} (0);
    \edge 0 3 \top;
    \draw[arc, bend right] (4) to node[left]{$a$} (3);
    \draw[arc, bend right] (3) to node[right]{$a$} (4);
    \node[state](5) at (0,0.5) {5};
    \node[state](6) at (1,0.5) {6};
    \node[state](7) at (2,0.5) {7};
    \node[state](8) at (2.4,1.3) {8};
    \node[state](9) at (1.6,1.3) {9};
    \inst 5; \fnst 7;
    \draw[arc, bend right] (5) to node[below]{$a$} (6);
    \draw[arc, bend right] (6) to node[above]{$a$} (5);
    \edgeb 6 7 \top;
    \draw[arc] (7) to node[right]{$a$} (8);
    \draw[arc] (8) to node[above]{$a$} (9);
    \draw[arc] (9) to node[left]{$a$} (7);
  \end{tikzpicture}
\end{align*}
The associated monoid is huge: it has $2^{10^2}$ elements; however, only the elements in the image of the homomorphism are relevant, and we shall see that it suffices to look at six of them.

Let $A\eqdef\Delta(a)$ and $T\eqdef\Delta(\top)$ be the transition relations for $a$ and $\top$.
Looking at $A$ and $T$ as 01-matrices, $T$ only contains two non-zero entries, at positions $(0,3)$ and $(6,7)$, and $A$ can be presented as a block-diagonal matrix:
\begin{align*}
  A\eqdef \left(
  \begin{array}{cccc}
    M&&&\\
     &N&&\\
     &&N&\\
     &&&M\\
  \end{array}
  \right)
  \qquad\text{where}\qquad
  M\eqdef \left(
  \begin{array}{ccc}
    0&1&0\\
    0&0&1\\
    1&0&0\\
  \end{array}
  \right)
  \qquad\text{and}\qquad
  N\eqdef \left(
  \begin{array}{cc}
    0&1\\
    1&0\\
  \end{array}
  \right)
\end{align*}
We have $M^3=1$ and $N^2=1$, so that $A^6=1$ and there are only six distinct matrices of the shape $A^i$.
The homomorphism $h=\hat\Delta$ maps words over $\set{a,\top}$ to 10x10 matrices.
Words of the shape $a^i$ are mapped to $A^i$,
words with at least two occurrences of $\top$ are mapped to the zero matrix, and words of the shape $a^i\top a^j$ are mapped to $A^iTA^j$.
Therefore, there are at most 6+1+36 elements in the image of $h$.

Let us now compute the predicate $P$ on these elements (here this predicate takes a single argument since there is only one atom). We write $i\equiv j[k]$ for $i$ equals $j$ modulo $k$):
\begin{align*}
  \begin{cases}
    P(A^i) = \text{false}\\
    P(0) = \text{false}\\
    P(A^iTA^j) = (i\equiv 0 [3] \land j\equiv 0 [2]) \lor (i\equiv 1 [2] \land j\equiv 0 [3])
  \end{cases}
\end{align*}
Now observe that $(XT)^*X=X+XTX$ when $X=A^i$, and $(XT)^*X=X$ in the other cases ($X=0$ or $X=A^iTA^j$). By Fact~\ref{fact:P'}, we have $P'(A^i)=P(A^i+A^iTA^i)$, and $P'(X)=P(X)$ in the other cases.
Further simplifying $P'(A^i)$, we get:
\begin{align*}
  P'(A^i) &= P(A^i+A^iTA^i)\\
          &= \text{false} \lor (i\equiv 0 [3] \land i\equiv 0 [2]) \lor (i\equiv 1 [2] \land i\equiv 0 [3])\\
          &= i\equiv 0 [3] \land (i\equiv 0 [2] \lor i\equiv 1 [2]) \\
          &= i\equiv 0 [3]
\end{align*}
In other words, with respect to the predicate $P$, there are only two elements in the image of $h$ that are added to obtain $P'$: $A^0$ and $A^3$.

We finally observe that $\ih(A^0)=[6^*]$ and $\ih(A^3)=[6^*3]$, so that
\begin{align*}
  [s(e)]=\ih(P')=\ih(P)\cup \ih(A^0)\cup \ih(A^3)=[e+6^*+6^*3]
\end{align*}

\subsection{A PSpace algorithm}
\label{ssec:alg}

Like for Theorem~\ref{thm:lang}, the proof of Theorem~\ref{thm:rel} leads to an algorithm for deciding the equational theory of \KATF. Indeed, we have $\KATF\proves e=f$ iff $[s(r(e))]=[s(r(f))]$, so that it suffices to be able to compare the latter guarded string languages.
This is possible since the functions $r$ and $s$ are computable.
However, while the function $r$ is linear, the function $s$ is costly: extracting a regular expression from a monoid (or similarly from a finite automaton) is exponential in general.

To obtain a \pspace algorithm, we avoid computing $s$ and work directly with recognisers.
Indeed, given an expression $e$ of size $n$, our constructions define a finite recogniser as in Figure~\ref{fig:rec}. This gives us a monoid for $E[e]$ whose elements are binary relations over $O(n)$ states\footnote{Assuming a regular-expression-to-automata function producing non-deterministic finite automata with linearly many states, as is usually the case~\cite{thompson68,Antimirov96}.}. In other words, elements are square 01-matrices of dimension $O(n)\times O(n)$. Those elements can be stored in quadratic space, and the various operations of the recogniser can be computed in polynomial time:
\begin{itemize}
\item the monoid product is nothing but matrix multiplication;
\item calling the homomorphism $h$ on a pair $(\alpha,a)$ requires a matrix multiplication and a reflexive-transitive closure;
\item testing whether a pair $(R,\alpha)$ is accepted requires three multiplications and two reflexive-transitive closures. (Note that the formula we use for $P$ in Figure~\ref{fig:rec} comes from Fact~\ref{fact:P'}.)
\end{itemize}
Putting everything together, we obtain the algorithm in Figure~\ref{fig:alg}.
This algorithm is non-deterministic: it progressively guesses a potential counter-example---a guarded string---and checks whether it is indeed a counter-example using recognisers for closed languages of guarded strings as deterministic guarded string automata.
This algorithm requires quadratic space: it stores only the two 01-matrices $x$ and $y$, whose respective dimensions are linear in the sizes of $r(e)$ and $r(f)$, and thus $e$ and $f$.

It may seem surprising that this algorithm has an endless loop and never returns \emph{true}.
Still, we can turn it into a (deterministic, terminating) \pspace algorithm by Savitch' theorem~\cite{Savitch}.
Intuitively, we can explore all non-deterministic choices and halt returning \emph{true} when all configurations (i.e., pairs $\tuple{x,y}$ of 01-matrices) have been visited and no counter-example was found.

%% asymmetric version of the algorithm, for inequations
% \begin{figure}[t]
%   \centering
% \begin{codeNT}
% // inputs two expressions $e,e'$; outputs false $ $ iff $[e]\not\subseteq E[r(e')]$
% $\tuple{X,I,\Delta,F}$ := NFA for $e$
% $\tuple{X',I',\Delta',F'}$ := NFA for $r(e')$
% guess $x\in I$
% $R$ := $1_{X'}$
% while true do
%   guess $\alpha\in\Atom$
%   if $x \in \Delta(\alpha)^*\cdot F$ and $I' \cdot \left(R\cdot \Delta'(\alpha)^*\cdot\Delta'(\top)\right)^*\cdot R\cdot \Delta'(\alpha)^* \cdot F' = \emptyset$
%   then return false
%   guess $a\in\Sigmat$
%   $x$ := guess in $\set{x}\cdot \Delta(\alpha)^*\cdot\Delta(a)$
%   $R$ := $R \cdot \Delta'(\alpha)^*\cdot\Delta'(a)$
% done
% \end{codeNT}
%   \caption{A \npspace algorithm for the inequational theory of \KATF.}
% \end{figure}

\begin{figure}[t]
  \centering
\begin{codeNT}
// inputs an expression $e$; outputs a recogniser for $E[e]$
$\tuple{X,I,\Delta,F}$ := non-deterministic finite automaton for $[e]$
$M$ := $\tuple{\pow(X^2),\cdot,1}$
$h(\alpha,a)$ := $\Delta(\alpha)^*\cdot\Delta(a)$
$P(R,\alpha)$ := $\left(R\cdot \Delta(\alpha)^*\cdot\Delta(\top)\right)^*\cdot R\cdot \Delta(\alpha)^* ~\cap~ I{\times}F \neq \emptyset$
return $\tuple{M,h,P}$
\end{codeNT}
  \caption{Recogniser for a language of the form $E[e]$.}
    \label{fig:rec}
\end{figure}

\begin{figure}[t]
  \centering
\begin{codeNT}
// inputs two expressions $e,f$; outputs false $ $ iff $E[r(e)]\neq E[r(f)]$
$\tuple{M,h,P}$ := recogniser for $E[r(e)]$
$\tuple{N,g,Q}$ := recogniser for $E[r(f)]$
$x$ := $1_{M}$
$y$ := $1_{N}$
while true do
  guess $\alpha\in\Atom$
  if $P(x,\alpha)\neq Q(y,\alpha)$ then return false
  guess $a\in\Sigmat$
  $x$ := $x \cdot_M h(\alpha,a)$
  $y$ := $y \cdot_N g(\alpha,a)$
done
\end{codeNT}
  \caption{Non-deterministic \pspace algorithm for the equational theory of \KATF.}
  \label{fig:alg}
\end{figure}

The equational theory of \KATF contains that of \KA, which amounts to language equivalence of regular expressions, which is \pspace-hard~\cite[Lemma~2.3]{MS72}\cite[Proposition~2.4]{HUNT1976}.
Therefore we deduce:
\begin{thm}
  \label{thm:katf:pspace}
  The equational theory of \KATF is {\rm\pspace}-complete.
\end{thm}

\section{Relations with a greatest element}
\label{sec:relp}

A \emph{generalised $\St$-algebra of relations} is an $S$-subalgebra $A$ of an algebra of relations such that $A$ has a greatest element, seen as an $\St$-algebra by using this greatest element for the constant $\top$. We write \RELP for the class of all generalised $\St$-algebras of relations.

Intuitively, \RELP consists of models of binary relations where $\top$ is not necessarily the full relation, only a greatest element.
As an example, consider relations $R$ over the natural numbers such that $i\leq j$ whenever $\rin R i j$. Those form an $S$-algebra with greatest element the order relation $\leq$ itself, which is not the full relation.

In the literature, \RELP is sometimes preferred over \REL because it is closed under taking subalgebras and products, and actually forms a quasivariety~\cite{AM11}. (In contrast, it is not clear whether \REL is closed under products: the two obvious ways of embedding a pair of relations into a new relation fail to preserve either union or top---\REL as defined here is not closed under taking subalgebras either, but defining it in such a way would not change the results from the present paper.)

The equational theory of \RELP differs from that of \REL. For instance, the previous example of ordered relations shows that $\RELP\not\models x\leq x\cdot\top\cdot x$. Indeed, for $x=\set{\tuple{0,1}}$, $x\cdot\top\cdot x$ is empty since $\top$ does not relate $1$ to $0$.

We show below that the equational theory of \RELP actually coincides with that of \GSL, and can thus be axiomatised by \KATT.

\begin{prop}
  \label{prop:embed}
  Every member of \GSL embeds into a member of \RELP.
\end{prop}
\begin{proof}
  We adapt the technique used by Pratt for Kleene algebras (without top)~\cite[third page]{pratt80:cayley} and later reused by Kozen and Smith for Kleene algebras with tests~\cite[Lemma~5]{kozens96:kat:completeness:decidability}.
  For a set $X$, let $M(X)$ be the set of relations $R$ on $\GS_X$ such that for all guarded strings $u,v$, $u$ is a prefix of $v$ whenever $\rin R u v$. The $S$-operations on relations restrict to $M(X)$, so that $M(X)$ is an $S$-algebra, and setting $\top\eqdef\set{\tuple {u, u\diamond v}\mid u,v \in \GS_X}$ turns it into a member of \RELP.
  We embed the member $\pow(\GS_X)$ of \GSL into $M(X)$ as follows:
  \begin{align*}
    \iota\colon &\pow(\GS_X) \to M(X)\\
           &L\mapsto\set{\tuple{u,u\diamond v}\mid u\in \GS_X,~v\in L}
  \end{align*}
  The function $\iota$ is easily shown to be an $\St$-algebra homomorphism, and it is injective (since, e.g., $L=\set{\alpha u\mid \tuple{\alpha,\alpha u}\in\iota(L)}$).
\end{proof}
Note that it is crucial that we consider \RELP rather than \REL here: the above construction would not give an $\St$-algebra homomorphism if we were not restricting to relations of a certain shape: $\top$ would not be preserved.

\begin{cor}
  \label{cor:relp}
  For all expressions, we have
  \begin{align*}
    \GSL\models e=f
    \iff
    \RELP\models e=f
    \iff
    \KATT\proves e=f
  \end{align*}
\end{cor}
\begin{proof}
  That $\RELP\models e=f$ entails $\GSL\models e=f$ is a direct consequence of Proposition~\ref{prop:embed}. That $\KATT\proves e=f$ entails $\RELP\models e=f$ follows from the soundness of $\KATT$ axioms w.r.t. \RELP. We conclude by Theorem~\ref{thm:lang}.
\end{proof}

Similarly to \RELP, we can define a class \GSLP of $\St$-algebras which is closed under taking subalgebras and where $\top$ is not necessarily the full language. However, unlike with \RELP and \REL, the equational theory of \GSLP coincides with that of \GSL (and \RELP). Indeed the axioms of \KATT remain sound for \GSLP.

\section{Conclusion}
\label{sec:ccl}

We have proved completeness of two axiomatic systems about regular expressions with tests and top, \KATT and \KATF, with respect to guarded string language models and relational models, respectively.
We have established that the corresponding equational theories are \pspace-complete, and that they can be reconciled by allowing relational models where top is only a maximal element, not necessarily the full relation.

For \KATF, we have proved a graph-theoretical characterisation of the equational theory of binary relations, we have established a relationship between this graph-theoretical characterisation and a notion of closed guarded string language, and we have used an extension of the theory of finite monoid recognition for guarded string languages.

\paragraph*{Related work}

Zhang et al. gave a completeness result for \KATT, in terms of guarded string languages~\cite[Theorem~9]{ZhangAG22}. They observed that this axiomatisation is incomplete for \REL, that it does not suffice to properly express \emph{incorrectness triples}, and they left the existence of a complete axiomatisation for relational models open.
Our Theorem~\ref{thm:rel} gives a positive answer to this question.

For the theory \KATT, the main completeness results of Zhang et al.~\cite[Theorems~7 and~9]{ZhangAG22} are wrong: the model of guarded strings they designed equates too many expressions (namely, $\Sigma^*$ and $\top$---see Remark~\ref{rem:bug}).
Our Theorem~\ref{thm:lang} uses a slightly different language model and yields a linear reduction from \KATT to \KAT, so that, e.g., \cite[Theorem~10]{ZhangAG22} about the complexity of \KATT remains true.

Zhang et al. also gave a completeness result w.r.t.\ generalised relational models~\cite[Theorem~8]{ZhangAG22}. Their proof is problematic because it relies on their Theorem~7, but the key idea remains valid: adapting Pratt's trick to embed language models into relational ones. We use the very same technique to obtain Corollary~\ref{cor:relp}.

\medskip

We recently proved a completeness result for \KATF~\cite[Section~7]{prw:hypotheses:journal}, w.r.t.\ a notion of closed language defined differently than in the present work. There the emphasis is on modularity, complexity aspects are not considered, and KAT with a top element is an example among others. The closed language model defined there is most probably equivalent to the one we use in the present paper (using arguments likes the one developed in~\cite[Appendix~C]{prw:hypotheses:journal} for plain KAT, which we would like to generalise in the future).
The two completeness proofs are rather different. The present one is more direct, uses guarded strings and finite monoids, and yields a \pspace algorithm. In contrast, the one in~\cite{prw:hypotheses:journal} avoids guarded strings but requires more general results about Kleene algebra with hypotheses, and does not give any reasonable algorithm. Our characterisations of the equational theory of \REL (Theorem~\ref{thm:rel:graphs}, Proposition~\ref{prop:rel}) also lie out of the scope of~\cite{prw:hypotheses:journal}.

\paragraph*{Future work}

A Hoare triple $\set\alpha e \set{\beta}$ for partial correctness can be encoded in KAT as an equation
$\alpha\cdot e\cdot\lnot\beta = 0$~\cite{kozen00:kat:hoare}. Since hypotheses of the more general shape $e=0$ can be incorporated into the equational theory of KAT~\cite{cohen94:ka:hypotheses,hardink02:kat:hypotheses}, one can automate reasoning about partial correctness~\cite{pous:itp13:ra}.

Zhang et al.~\cite{ZhangAG22} have shown how to encode an incorrectness triple $[\alpha]e[\beta]$ as an inequation $\beta \leq \top\cdot \alpha\cdot e$. A natural question is whether such hypotheses can also be eliminated in \KATF, in order to automate reasoning about incorrectness triples.
The modular tools we developed in~\cite{prw:ramics21:mkah,prw:hypotheses:journal} could prove useful, provided we find a way to extract efficient algorithms from the resulting reductions.

\paragraph*{Acknowledgements}

We would like to thank Paul Brunet, Amina Doumane, and Jurriaan Rot for the discussions that eventually led to this work, and the CONCUR'22 reviewers for all their comments.
We also thank Denis Kuperberg for showing us how the roots of a regular language could be computed via finite monoids, which is the key idea that led to the construction in Section~\ref{ssec:rel:compl:clang}.

\bibliographystyle{alphaurl}
\bibliography{pous,main}

\clearpage
\appendix
\section{Proof of Theorem~\ref{thm:rel:graphs}}
\label{app:graphs}

We give here a proof of Theorem~\ref{thm:rel:graphs}.
Variants of this theorem appeared for Kleene allegories without top in~\cite[Theorem~6]{bp:lics15:paka}, and for Kleene allegories with top in~\cite[Theorem~16]{pous:stacs18}.

First we observe that valuations into relational models are very close to (potentially infinite) graphs in the sense of Definition~\ref{def:graph}: it suffices to adjoin to them an input and an output.
\begin{defi}[Graph of a valuation]
  \label{def:graph:val}
  Consider a member of \REL: relations on some set $X$ with a function $p\colon X\to\Atom$.
  Let $\sigma\colon\Sigma\to\rel X$ be a valuation of $\Sigma$ into relations on $X$.
  For all elements $i,j\in X$, we define the graph $\tuple{\sigma,i,j}\eqdef\tuple{X,F,p,i,j}$ where
  $F\eqdef\set{\tuple{x,a,y}\mid a\in\Sigma,~\tuple{x,y}\in \sigma(a)}$.
\end{defi}

The first key lemma characterises evaluation of expressions not using $0,+,\cdot^*$ in a relational model, in terms of graph homomorphisms.
In our case, expressions not using $0,+,\cdot^*$ can be represented by guarded strings.
Such a lemma appeared first in \cite[Lemma~3]{AB95} for a signature including intersection and converse, but not top nor tests. Under its original formulation, its extension to cover top is trivial once we realise that the graph of $\top$ should simply be a graph without edges and exactly two vertices (the input and the output).
\begin{lem}
  \label{lem:bredikhin}
  Let $\sigma\colon\Sigma\to\rel X$ be a valuation of $\Sigma$ into a member of \REL.
  For all guarded strings $u$, we have
  \begin{align*}
    \tuple{i,j}\in\hat\sigma(u) \iff \tuple{\sigma,i,j}\lhd g(u)
  \end{align*}
\end{lem}
\begin{proof}
  By induction on $u$.
  \begin{itemize}
  \item if $u=\alpha$ is an atom, then both sides reduce to the condition $i=j \land p(i)=\alpha$;
  \item if $u=\alpha a \beta$ has length one
    \begin{itemize}
    \item if $a=\top$, then both sides reduce to the condition $p(i)=\alpha \land p(j)=\beta$;
    \item if $a\in\Sigma$, then both sides reduce to the conjunction of the previous condition and $\tuple{i,j}\in\sigma(a)$;
    \end{itemize}
  \item if $u=v\alpha w$ for two smaller guarded strings $v\alpha$ and $\alpha w$ then we have
    \begin{align*}
      &\tuple{i,j}\in\hat\sigma(v\alpha w)\\
        \tag{by definition}
      \Leftrightarrow~& \exists k,~\tuple{i,k}\in\hat\sigma(v)~\land~p(k)=\alpha~\land~\tuple{k,j}\in\hat\sigma(\alpha w)\\
      \Leftrightarrow~& \exists k,~\tuple{i,k}\in\hat\sigma(v\alpha)~\land~\tuple{k,j}\in\hat\sigma(\alpha w)\\
        \tag{by induction hypothesis on $v\alpha$ and $\alpha w$}
      \Leftrightarrow~& \exists k,~\tuple{\sigma,i,k}\lhd g(v\alpha)~\land~\tuple{\sigma,k,j}\lhd g(\alpha w)\\
      \Leftrightarrow~& \tuple{\sigma,i,j}\lhd g(v\alpha w)
    \end{align*}
    (The last equivalence comes from a simple analysis of the homomorphisms whose source is a sequential composition of two graphs---see, e.g., \cite[Lemma~2(ii)]{AB95}.)\qedhere
  \end{itemize}
\end{proof}

The second key lemma characterises the evaluation of an arbitrary expression in terms of (the evaluations of) the guarded strings in the language of that expression.
Variants of such a lemma often appear in the literature for \emph{star-continuous} models, rather than just relational ones (e.g., \cite[Lemma~4]{kozens96:kat:completeness:decidability}).
% in that ref, Dexter points to \cite[Lemma~7.1, page~35]{K91a}, but I don't have the book...
\begin{lem}
  \label{lem:expr:words}
  Let $\sigma\colon\Sigma\to\rel X$ be a valuation of $\Sigma$ into a member of \REL.
  For all expressions $e$, we have
  \begin{align*}
    \hat\sigma(e)=\bigcup_{u\in [e]}\hat\sigma(u)
  \end{align*}
\end{lem}
\begin{proof}
  By an easy induction on $e$, using distributivity of $\cdot$ over arbitrary unions in \REL.
\end{proof}

Equipped with those two lemmas, we obtain the announced theorem.
\begin{thm}
  \label{thm:rel:clang}
  For all expressions $e,f$, we have:
  \begin{align*}
    \REL\models e\leq f \iff \forall u\in[e],~\exists v\in[f],~g(u)\lhd g(v)
  \end{align*}
\end{thm}
\begin{proof}
  For the forward implication, assume $\REL\models e\leq f$ and let $u\in [e]$.
  Let $n$ be the length of $u$ and consider relations on $[0;n]$, a member of \REL with the function $p$ mapping $i\leq n$ to the $i$th atom in $u$. Define $\sigma\colon\Sigma\to\rel{[0;n]}$ by $\tuple{i,j}\in\sigma(a)$ if the $i$-th letter of $u$ is $a$ and $j=i+1$. The graph $g(u)$ is nothing but $\tuple{\sigma,0,n}$, so that
  we have $\tuple{0,n}\in\hat\sigma(u)$ by Lemma~\ref{lem:bredikhin}, using the identity graph homomorphism. Thus we consecutively get $\tuple{0,n}\in\hat\sigma(e)$ by Lemma~\ref{lem:expr:words},  $\tuple{0,n}\in\hat\sigma(f)$ by assumption, and $\tuple{0,n}\in\hat\sigma(v)$ for some $v\in [f]$ by
  Lemma~\ref{lem:expr:words} again.
  Lemma~\ref{lem:bredikhin} finally gives $g(u)=\tuple{\sigma,0,n}\lhd g(v)$, as required.

  For the backward implication, assume the right-hand side and let $\sigma\colon\Sigma\to\rel X$ be a valuation into a member of \REL.
  For all $i,j\in X$, we have
  \begin{align*}
    &\tuple{i,j}\in\hat\sigma(e)\\
    \tag{by Lemma~\ref{lem:expr:words}}
    \Leftrightarrow~&\tuple{i,j}\in\hat\sigma(u)\text{ for some }u\in[e]\\
    \tag{by Lemma~\ref{lem:bredikhin}}
    \Leftrightarrow~& \tuple{\sigma,i,j}\lhd g(u)\text{ for some }u\in[e]\\
    \tag{by assumption}
    \Rightarrow~& \tuple{\sigma,i,j}\lhd g(u)\text{ for some }u,v\text{ s.t. }v\in[f]\text{ and }g(u)\lhd g(v)\\
    \tag{by transitivity of $\lhd$}
    \Rightarrow~& \tuple{\sigma,i,j}\lhd g(v)\text{ for some }v\in[f]\\
    \tag{by Lemma~\ref{lem:bredikhin}}
    \Leftrightarrow~&\tuple{i,j}\in\hat\sigma(v)\text{ for some }v\in[f]\\
    \tag{by Lemma~\ref{lem:expr:words}}
    \Leftrightarrow~& \tuple{i,j}\in\hat\sigma(f)
  \end{align*}
  Whence $\hat\sigma(e)\subseteq\hat\sigma(f)$, and thus $\REL\models e\leq f$ as required.
\end{proof}

\end{document}